\documentclass{article}  
\usepackage[utf8]{inputenc}
\usepackage[margin=1.0in]{geometry}

\usepackage[noblocks]{authblk}%affiliation,author
\usepackage[english]{babel}
\usepackage{braket}
\usepackage{amsthm}
\usepackage{cancel,amsmath, amsfonts, amssymb, graphicx, booktabs, braket, bbold}
\usepackage[all]{xy}
\usepackage{tikz-cd} %commuting diagrams pack
\usepackage{framed} %paragraph's frame
\usepackage{subfigure}
\usepackage{qcircuit}
\usepackage{algorithm}
\usepackage[backgroundcolor=green]{todonotes}

\newcommand{\indep}{\mathbin{\perp\!\!\!\perp}}
\newcommand{\tr}{\mbox{\textnormal{Tr}}}

\newtheorem{lemma}{Lemma}[section]
\newtheorem{assumption}{Assumption}[section]
\newtheorem{proposition}{Proposition}[section]
\newtheorem{theorem}{Theorem}[section]
\newtheorem{problem}{Problem}[section]
\newtheorem{definition}{Definition}[section]
\newtheorem{corollary}{Corollary}[section]
\newcommand{\id}{\mbox{id}}
\newcommand{\liouville}[1] { \mathcal{B}\left(\mathcal{H}_{#1}\right)}
\newcommand{\hilbert}[1] { \mathcal{H}_{#1}}

\newcommand{\qI}[1]{I_{{}_\rho}\left(#1 \right) }

\newcommand{\pTr}[2]{\text{Tr}_{#1}\left[#2\right] }

\newcommand{\Petz}[3]{\rho_{#1}^{\frac{1}{2}}\rho_{#2}^{-\frac{1}{2}}\rho_{#3}\rho_{#2}^{-\frac{1}{2}}\rho_{#1}^{\frac{1}{2}}}

\newcommand{\norm}[1]{\left\|\,#1\,\right\|}

\newcommand{\trnorm}[1]{\norm{#1}_{\rm tr}}
\newcommand{\trn}[1]{\trnorm{#1}}
\newcommand{\ketbra}[2]{\ket{#1}\!\bra{#2}}

\newcommand{\tensor}{\otimes}

\newcommand{\half}{{1 \over 2}}

\newcommand{\QSD}{\mbox{QSD}}

\newcommand{\poly}{\mbox{poly}}

\newcommand{\cp}[1]{\phi\left(\rho_{#1} \right) }

\title{On the complexity of finding the maximum entropy compatible quantum state}
\author[1,2]{S. Di Giorgio}
\author[1,2]{P. Mateus}
\affil[1]{Departamento de Matemática, Instituto Superior Técnico, Universidade de Lisboa, Av. Rovisco Pais, 1049-001 Lisboa, Portugal}
\affil[2]{Instituto de Telecomunica\c{c}\~oes, 1049-001 Lisboa, Portugal}
\date{}

\begin{document}

\maketitle

\begin{abstract}
Herein we study the problem of recovering a density operator from a set of compatible marginals, motivated from limitations of physical observations. Given that the set of compatible density operators is not singular, we adopt Jaynes' principle and wish to characterize a compatible density operator with maximum entropy. We first show that comparing the entropy of compatible density operators is QSZK-complete, even for the simplest case of 3-chains. Then, we focus on the particular case of quantum Markov chains and trees and establish that for these cases, there exists a quantum polynomial circuit that constructs the maximum entropy compatible density operator. Finally, we extend the Chow-Liu algorithm to the same subclass of quantum states.
%Finally, we provide an extension of the Chow-Liu algorithm for a subclass of quantum states
\end{abstract}

\section{Introduction}

Determining whether a set of marginal density operators has a global density operator compatible with them has been a hot research subject in the area of quantum information and mathematical physics~\cite{kly:06,liu:06}. For the moment, there is no efficient or closed-form way to detect whether there exists such a global density operator, except for some particular cases~\cite{hig:ats:sud:szu:03,bra:ser:03,hub:2018}. In this paper, we focus on a different problem. We are given a promise that the marginals have a compatible global state and want to infer the global according to Jaynes principle~\cite{jay:57}. In physics, this corresponds to perform quantum tomography on the marginals, and then infer the global state using maximum entropy principle. Indeed, we might not be able to have the physical apparatus or the computational resources to perform the tomography over the global state, but rather, only over partial states.  

As a consequence, the following question arises: can we infer efficiently (even with a quantum computer) the global state from the marginals? Furthermore, can we infer the state which corresponds to performing fewer assumptions? The latter state is usually considered to be the state with the highest entropy. 

\paragraph{Contributions of the paper.}
We start by obtaining a negative result by showing that comparing the entropies of 3-chains (the simplest non-trivial scenario) is QSZK-complete~\cite{wat:08,ber:vaz:97,wat:02}. This result hints that to find the maximum entropy compatible state given, at least, two marginals should be not feasible, even for a quantum computer~\cite{aro:bar:09} (at least by performing some entropy-monotonic step-by-step optimization).

Next, we proceed to restrict the class of quantum states to make the problem feasible for such a class. We consider quantum Markov trees, states for which each 3-subchain form a quantum Markov chain~\cite{faw:ren:15}. In this case, we show that the maximum entropy compatible problem is in P,  and that there exists a polynomial-time quantum circuit that constructs the maximal entropy compatible state. Finally, we use this result to extend the Chow-Liu algorithm~\cite{cho:liu:68} for quantum states whose all 3-subchains are quantum Markov chains.

\paragraph{Organization of the paper.} In Section 2 we give some background and state clearly the problems we are addressing. In Section 3, we attain the hardness of comparing the entropy of a compatible chain. In Section 4, we consider the restriction of the maximum entropy problem to quantum Markov trees. There, we provide the polynomial-time solution for this case, how to construct the solution with a polynomial-quantum circuit and the generalization of Chow-Liu algorithm. Some of the proofs are left to the appendices. Finally, we draw some conclusions and leave some open problems in Section 5. 

\section{Background and problem statement}

Throughout this work we assume all quantum states and operators to be defined over a finite dimensional Hilbert space $\mathcal{H}$ that is composed of $n$ parts, such that $\mathcal{H}=\otimes_{i=1}^n\mathcal{H}_i$.
We denote by $\mathcal{I}$ a collection of subsets of $\{1,\dots n\}$, that is a set of subsets of $\{1\dots n\}$, and throughout the text we call $\mathcal{I}$ the {\em set of marginals indexes}. Elements of $\mathcal{I}$ are denoted by $J$, and its complement is represented by $\overline{J}$.
Given $\mathcal{I}$, we are interested in density operators that are compatible with a $\mathcal{I}$-indexed family of marginal density operators $\mathcal{C}$ where 
$\mathcal{C}=\{\rho_J\in\liouville{J}\}_{ J\in\mathcal{I}}$ such that  
\begin{equation}
\pTr{\overline{J\cap J'}}{\rho_J}=\pTr{\overline{J\cap J'}}{\rho_{J'}}\textrm{ for all }J,J'\in \mathcal{I},
\end{equation}
where $\mathcal{H}_{J}=\bigotimes_{i\in J}\mathcal{H}_i$. We call each element $\rho_J$ a {\em marginal density operator}. We also denote by  $\mathcal{Q}(\mathcal{C})=\{Q_J\}_{J\in \mathcal{I}}$ a family of quantum circuits such that $Q_J$ constructs the density operator $\rho_J$.

The {\em compatibility set} $\mbox{Comp}\left( \mathcal{C}\right)$ associated to a given family of compatible marginals $\mathcal{C}$ is the set of density operators over $\mathcal{H}$ that admits as partial traces all the elements of $\mathcal{C}$, that is:
\begin{equation}
\mbox{Comp}\left( \mathcal{C}\right):=\left\lbrace \rho\in\liouville{} :\, \pTr{\overline{J}}{\rho}=\rho_{J}\ \textrm{ for all } J\in\mathcal{I}\right\rbrace.
		\end{equation}
The family $\mathcal{C}$ is said to be {\em admissible} when $\mbox{Comp}\left( \mathcal{C}\right)\neq 0$, that is, if it admits at least one density operator whose marginals coincide with those in $\mathcal{C}$.

We start by noticing that, compatible sets, where all marginal density operators are diagonal for the same basis (that is, represent discrete probability distributions), collapse in the classical compatible marginal problem~\cite{yu:04}. This classical problem has been shown to be NP-complete for the three-dimensional case~\cite{loe:04}. There are many cases for which it is solvable~\cite{fri:13}, and there is always a solution if we consider only two-body marginals (bipartite marginals) that form an acyclic graph. 

The relevant case where the marginals are not diagonal for the same basis has been the target of several research works and is called the quantum compatible marginal problem. Liu showed that this problem is QMA-complete~\cite{liu:06}.\\

\begin{framed}
\begin{problem}\em 
 Quantum Compatible Marginal Problem (QCMP): \begin{itemize}
	\item \texttt{Input:} A family of  circuits $\mathcal{Q}(\mathcal{C})$ that construct the family of marginal density operators $\mathcal{C}$.
	\item \texttt{Accept:} if $\mathcal{C}$ is admissible.
    \item \texttt{Reject:} if $\mathcal{C}$ is not admissible.
\end{itemize}
\end{problem}
\end{framed}
\vspace*{3mm}

In some cases we know that $\mathcal{C}$ is admissible, for instance when we are promised that the marginals $\rho_J$ are indeed partial traces of a global state. In Physics, it is reasonable to assume that we can prepare many copies of a global system, but in general, we can only partially observe it. In this case, given that we have many copies of the global system, we would be able to characterize in full detail the partial traces and know that they form an admissible set. The question now is to infer the global state with maximum entropy among those in the compatibility set. This leads to the following problem.\\

\begin{framed}
\begin{problem}\em 
 Maximum Entropy Compatible Marginal Problem (MECMP):\begin{itemize}
	\item \texttt{Input:} A family of  circuits $\mathcal{Q}(\mathcal{C})$ promised to construct an admissible $\mathcal{C}$, and a real value $k$.
	\item \texttt{Accept:} if there exists a $\rho\in \mbox{Comp}\left( \mathcal{C}\right)$ such that $S(\rho)\geq k$
    \item \texttt{Reject:} otherwise.
\end{itemize}
\end{problem}
\end{framed}

Given the general complexity of this problem, we focus on the more straightforward case where all sets $J$ in $\mathcal{I}$ have two indexes. Thus we consider that we are given a set of compatible two-body marginals, and we want to reconstruct the maximum entropy state compatible with those marginals. For this two-body case, it is possible to construct an associated graph, where each two-body marginal denotes an edge.

\begin{definition}
	\label{def: associate graph}
\emph{	Let $\mathcal{C}$ be a $\mathcal{I}$-indexed family of two-body compatible marginal density operators. The {\em associated graph} $\mathcal{G}_\mathcal{C}$ is $\left(\{1,\dots, n\}, E\right)$, where $(i,j)\in E$ if $\{i,j\}\in \mathcal{I}$.}
\end{definition}

In the simplest non-trivial case we have that $n=3$ and $\mathcal{I}=\{\{1,2\},\{2,3\}\}$. We call this case a 3-chain. In the next section, we show that given two density operators $\rho_0$ and $\rho_1$ in the compatible set of a 3-chain, comparing who has higher entropy is QSZK-complete. We denote the subspaces $\mathcal{H}_1$, $\mathcal{H}_2$ and $\mathcal{H}_3$ by $\mathcal{H}_A$, $\mathcal{H}_B$ and $\mathcal{H}_C$, respectively.

\section{Hardness of comparing entropy of a compatible chain}

Ben-Aroya et al.~\cite{aro:tas:07} showed that, given two quantum circuits $Q_0$ and $Q_1$ that generate two mixed states $\rho_0$ and $\rho_1$, respectively, such that $|S(\rho_0)-S(\rho_1)|>\frac{1}{2}$, determining whether $S(\rho_0)>S(\rho_1)$ is QSZK-complete. Thus, they conclude that it is quite improbable that computing the von Neumann entropy of a mixed state can be done in BQP~\cite{aar:05}. We further look into this problem by restricting to the case when $\rho_0$ and $\rho_1$ live in the same Hilbert space and have the same marginals. 
We state our problem as follows:
\begin{framed}
\begin{problem}\em 
 3-Chain Compatible Quantum Entropy Difference (3cQED):
\begin{itemize}
	\item \texttt{Input:} Two quantum circuits $Q_0$ and $Q_1$  that generate tripartite density operators $\rho_0$ and $\rho_1$, respectively, over the same Hilbert space of the form $\mathcal{H}_A\otimes \mathcal{H}_B \otimes \mathcal{H}_C$, promised that:
\begin{itemize}
    \item $Tr_{A}(\rho_0)=Tr_{A}(\rho_1)$;
    \item $Tr_{C}(\rho_0)=Tr_{C}(\rho_1)$;
    \item $|S(\rho_0)-S(\rho_1)|\geq 1/2$;
\end{itemize}	
    then,
	\item \texttt{Accept:} if $S(\rho_0)-S(\rho_1)\geq 1/2$;
    \item \texttt{Reject:} if $S(\rho_1)-S(\rho_0)\geq 1/2$.
\end{itemize}
\end{problem}
\end{framed}

Clearly, 3cQED is a particular case of QED, wherein the latter the Hilbert space of $\rho_0$ and $\rho_1$ does not have to be the same, nor the densities need to be tripartite.

Obviously, 3cQED is reducible to QED, and therefore it relies in QSZK. It remains to show that it is QSZK hard. To do so, we adapt the proof of Ben-Aroya et al., and reduce QSD$_{\alpha,\beta}$ (a well-known problem that is QSZK-complete) to 3cQED, for $0\leq\alpha\leq\beta^2\leq 1$. 

\begin{framed}
\begin{problem}\em Quantum state distance  (QSD$_{\alpha,\beta}$) with $0\leq \alpha^2\leq \beta \leq 1$:
\begin{itemize}
\item \texttt{Input:} Two quantum circuits $Q_0$ and $Q_1$ that prepare the states $\rho_0$ or $\rho_1$ promised that  
\begin{itemize}
\item either $||\rho_0-\rho_1||_{tr}\geq \beta$;
    \item  or $||\rho_0-\rho_1||_{tr}\leq \alpha$;
\end{itemize}
then, 
	\item \texttt{Accept:} $||\rho_0-\rho_1||_{tr}\geq \beta$,
	\item \texttt{Reject:} $||\rho_0-\rho_1||_{tr}\leq \alpha$. 
\end{itemize}
\end{problem}
\end{framed}
\vspace*{2mm}

In Problem~4, $||\rho_0-\rho_1||_{tr}$ denotes the trace distance between the operators $\rho_0$ and $\rho_1$.

%The reduction of QSD$_{\alpha,\beta}$ to 3cQED can be stated as the following problem:
%\paragraph{Problem} Consider $\rho_0$ and $\rho_1$, two density operators over the same Hilbert space $H$, obtained from circuits $Q_0$ and $Q_1$, such that
%\begin{itemize}
%\item either $||\rho_0-\rho_1||_{tr}\geq \beta$
%    \item  or $||\rho_0-\rho_1||_{tr}\leq \alpha$
%\end{itemize}
%Construct in polynomial-time  two tripartite states $\rho'$ and $\rho''$ (over some space depending on $H$) such that:
%\begin{itemize}
%	\item Two marginals must coincide (e.g. over the subsystems $\{A,B\}$ and $\{B,C\}$) 
%	\item $|S(\rho')-S(\rho'')|\geq 1/2$
%	\item $||\rho_0-\rho_1||_{tr}\geq \beta$ iff $S(\rho')>S(\rho'')$
%\end{itemize}

\begin{theorem}\label{thm1}
\em For any $0 \le \alpha < \beta^2\le 1$, $\QSD_{\alpha, \beta}$ is reducible to
3cQED.
\end{theorem}
 
%Because of time I did not adapted the wording from Ben Aroya et al, we need to rephrase this properly...

\begin{proof}
Given circuits $Q_0, Q_1$, that construct $\rho_0$ and $\rho_1$, we first apply the polarization lemma
(Lemma \ref{thm:polarization} in Appendix~\ref{ap1}) with $n=m_0$ and obtain circuits $R_0$ and $R_1$ that output density operators
$\mu_0, \mu_1$, respectively. We then construct two circuits $Z_0$ and $Z_1$ as
follows. $Z_1$ is implemented by a circuit which first applies a
Hadamard gate on a single qubit $b$, measures $b$ and then
conditioned on the result it applies either $R_0$ or $R_1$. The
output of $Z_1$ is $\xi_1=\half \ketbra{0}{0} \tensor \mu_0 + \half
\ketbra{1}{1} \tensor \mu_1$. Since we need to construct a tripartite system, we introduce the notation $\xi_1^{AC}$ to point out that the qubit part of $\xi_1$ belongs to system $A$ and the remaining part belongs to system $C$. As expected, $\xi_1^{CA}$ indicates that the qubit belongs to $C$ and remaining part to $A$.
Circuit $Z_0$ is the same as $Z_1$ except
that $b$ is traced out. The output of $Z_0$ is $\xi_0=\half \mu_0 +
\half \mu_1$. We shall denote by $\xi_0^A$ and $\xi_0^C$ if the state is in $A$ or $C$ subsystem, respectively. 

Finally, we denote by $\ket{\phi^\pm}^{AC}$ two maximally entangled states between $A$ and $C$. Moreover, we take $\zeta=\half \ket{\phi^+}\bra{\phi^+}+\half\ket{\phi^-}\bra{\phi^-}$ and
note that $S(\zeta)=1$. We denote by $Q$ the circuit that prepares $\zeta$.
%See Figure \ref{fig:??}.
Consider 
\begin{itemize}
 \item $\rho'=\xi_0^{A}\otimes \zeta^{AC}\otimes \xi_0^{C}\otimes \ket{0}\bra{0}_B$ ;
  \item $\rho''=\xi_1^{AC}\otimes \xi_1^{CA}\otimes \ket{0}\bra{0}_B$ .
\end{itemize}
Note that in $\rho'$ the subsystem of $A$ contains $\xi_0^{A}$ and a qubit of $\zeta^{AC}$; the subsystem of $C$ contains $\xi_0^{C}$ and the other qubit of $\zeta^{AC}$. Moreover, in $\rho''$, the subsystem of $A$ has a qubit entangled with $\mu_0$ and $\mu_1$ in the subsystem $C$ ($\xi_1^{AC}$); and has another $\mu_0$ and $\mu_1$ entangled with a qubit of $C$ ($\xi_1^{CA}$).

The reduction outputs the following pair of d.o. $(\rho', \rho'')$ together with the circuits that construct them, namely $Z_0 \tensor
Z_0 \tensor Q$ and  $Z_1 \tensor Z_1$. We ignore the construction of the state $\ket{0}\bra{0}_B$, which is trivial.

Start by observing that by tracing $C$ from both $\rho'$ and $\rho''$ we obtain $(\half\ket{0}\bra{0}+\half\ket{1}\bra{1})\tensor (\half\mu_0+\half\mu_1)\tensor \ket{0}\bra{0}$. The same state will be obtained by tracing subsystem $A$ from both $\rho'$ and $\rho''$. So, $\rho'$ and $\rho''$ have compatible marginals. 

\paragraph{Part 1}
If $(Q_0, Q_1) \in (\QSD_{\alpha, \beta})_{NO}$ then $(Z_0 \tensor
Z_0 \tensor Q, Z_1 \tensor Z_1) \in$ 3cQED$_{NO}$.
\ \\[2mm]
We know that $\trnorm{\rho_0 - \rho_1} \le \alpha$. By using the
Polarization lemma (Lemma~\ref{thm:polarization} in Appendix~\ref{ap1}) we get
$\trnorm{\mu_0 - \mu_1} \le 2^{-m_0}$. By the joint-entropy
theorem (Lemma \ref{Lemma:joint-entropy-theorem}),

$$S(\xi_1) = \half (S(\mu_0)+ S(\mu_1))~+~1.$$
\noindent
On the other hand, $\xi_0$ is very close both to $\mu_0$ and
to $\mu_1$. Specifically, $\trnorm{\xi_0 - \mu_1} =
\trnorm{\half \mu_0 - \half \mu_1} \le 2^{-m_0}$. Thus,
by Fannes' inequality (Lemma~\ref{Lemma:Fannes} in Appendix~\ref{ap1}) $|S(\xi_0) -
S(\mu_1)| \le 2^{-m_0} \cdot \poly(m_0) \leq 0.1~$, for large
enough $m_0$. Similarly, $|S(\xi_0) - S(\mu_0)| \le 0.1$. It
follows that

$$|S(\xi_0) - \half (S(\mu_0)+ S(\mu_1))| \le 0.1.$$

Combining the two equations we get $ S(\xi_1) - S(\xi_0) \ge
0.9$. Thus, $S(\rho'') - S(\rho') \ge 2*0.9 - 1 = 0.8$. Therefore, $(Z_0 \tensor Z_0
\tensor Q, Z_1 \tensor Z_1) \in$ 3cQED$_{NO}$.

\paragraph{Part 2}
If $(Q_0, Q_1) \in (\QSD_{\alpha, \beta})_{YES}$ then $(Z_0 \tensor
Z_0 \tensor Q, Z_1 \tensor Z_1) \in$ 3cQED$_{YES}$.
\ \\[2mm]
By the Polarization lemma (Lemma \ref{thm:polarization} in Appendix~\ref{ap1})
$\trnorm{\mu_0 - \mu_1} \ge 1 - 2^{-m_0}$. Using Lemma \ref{lem:ANTV-with-trace-norm} (in Appendix~\ref{ap1}) we get that 
$$S(\xi_0) \ge
\half[S(\mu_0) + S(\mu_1)] + 1-H(\half + \frac{\trn{\mu_0 -
\mu_1}}{2}) \ge \half[S(\mu_0) + S(\mu_1)] + 1 - H(2^{-m_0}).$$

By Lemma \ref{Lemma:joint-entropy-theorem} (in Appedix~\ref{ap1}) we know that $S(\xi_1) =
\half (S(\mu_0) + S(\mu_1)) + 1$. Therefore, for sufficiently large $m_0$ we have $S(\xi_1) -
S(\xi_0) = H(2^{-m_0}) < 0.1$.

In particular, $S(\rho'') - S(\rho') \le 2*0.1 - 1 = -0.8$ and $(Z_0 \tensor Z_0 \tensor Q,
Z_1 \tensor Z_1) \in$ 3cQED$_{YES}$.
\end{proof}

It follows that comparing the entropy of a set of compatible marginals is QSZK-complete, as this problem is also an instance of QED. As a consequence, we expect that finding the maximum entropy state is also hard. This fact does not imply that, given a state $\rho$ and the marginals, it is no possible to detect efficiently whether this state is one with maximum entropy among those compatible with the marginals. We now focus our attention on a particular case in which this problem can be addressed.

\section{Quantum Markov chains and trees}

Given that the general problem of finding the maximum entropy state is hard, even for 3-chains, we consider a simpler case. We focus on quantum Markov chains~\cite{sut:18} that rely on the Hilbert space $\mathcal{H}=\mathcal{H}_{A}\otimes\mathcal{H}_{B}\otimes\mathcal{H}_{C}$ and take $\mathcal{C}=\{\{A,B\},\{B,C\}\}$. To simplify notation, we drop the brackets and commas in the indexes and so, for instance, the partial trace $\rho_{\{A,B\}}$ is just denoted by $\rho_{AB}$ (the same simplification is applied for the Hilbert subspaces $\mathcal{H}_{\{A,B\}}$, which are denoted just by $\mathcal{H}_{AB}$).

Recall the definition of quantum Markov chain:
\begin{definition}
\label{def:QMC} \em \cite{sut:faw:ren:16}
A quantum Markov chain (QMC) is a 3-chain $A-B-C$ for which there exists a recovery map $\mathcal{R}_{B\to BC}:\mathcal{B}(\mathcal{H}_{B})\to \mathcal{B}(\mathcal{H}_{BC})$, i.e. an arbitrary trace-preserving completely positive (CPTP) map~(see, for instance,~\cite{cho:75,nie:chu:12}), s.t. $\rho_{ABC}=\left(\mathcal{I}_{A}\otimes\mathcal{R}_{B\to BC}\right)(\rho_{AB})$,
where $\mathcal{I}_{A}$ denotes the identity map on $\mathcal{B}(\mathcal{H}_A)$.
\end{definition}

By definition, the recovery map must fulfill that $\mathcal{R}_{B\to BC}\left( \rho_B\right) = \rho_{BC}$.
 
 \begin{definition}\label{def:PTQMC}
\em A family of QMC's $\{\rho^{(n)}_{ABC}\}_{n\in \mathbb{N}}$ is said to be {\em constructed in polynomial time} if all elements $\rho^{(n)}_{ABC}$ rely in the same (finite) Hilbert space $\mathcal{H}_{A}\otimes\mathcal{H}_{B}\otimes\mathcal{H}_{C}$ (that does not depend on $n$) and there is polynomial-time family of quantum circuits that generate both $\rho^{(n)}_{AB}$ and $\mathcal{R}^{(n)}_{B\to BC}$. 
 \end{definition}
 
 Given that the dimension of (a polynomial-time) quantum Markov chain does not grow with $n$, it can be represented in matrix form in polynomial-time by multiplying all the gates involved in the circuits that generate $\rho^{(n)}_{AB}$ and $\mathcal{R}^{(n)}_{B\to BC}$. We stress that to design circuits for density operators and CPTP maps we require only an ancilla space of the same dimension of the support of these operators/maps~\cite{aha:kit:nis:08}, and therefore the number of gates is polynomial in $n$, but the full dimension of the space (including ancillae) does not grow with $n$.
 
 From this point on, we assume that $\rho_{ABC}$ is invertible (on its support), as invertible density operators are dense. To derive the main result of the paper, we need to establish a central lemma relating quantum Markov chains with the Petz recovery map together with the strong subadditivity of von Neumann entropy. We give the proof in Appendix~\ref{appendix:proof:lemma6}.

\begin{lemma}\label{lemma: QMC - I(A:C|B)}\em Let $\rho_{ABC}$ be an invertible density operator. The following four assertions are equivalent:
	\begin{enumerate}
		\item $\rho_{ABC}$ is a QMC over the chain $A-B-C$.
		\item $ I_{\rho}(A:C|B)=0$, where $I_{\rho}(A:C|B):=S(\rho_{AB})+S(\rho_{BC})-S(\rho_{B})-S(\rho_{ABC})$.
		\item $\mathcal{P}_{B\to BC}(X):=\rho_{BC}^{\frac{1}{2}}((\rho_B^{-\frac{1}{2}}X \rho_B^{-\frac{1}{2}})\otimes \id_{C})\rho_{BC}^{\frac{1}{2}}, \text{ is a CPTP map for any } X\in \mathcal{B}(\mathcal{H}_{B})$ and preserves the partial trace $\rho_{AB}$.
		\item $\log\rho_{ABC}-(\log\rho_{AB})\otimes \id_{C}=\id_{A}\otimes (\log\rho_{BC})-\id_{A}\otimes(\log\rho_B)\otimes \id_{C}$.
		\end{enumerate}
\end{lemma}

The map $\mathcal{P}_{B\to BC}(X)$ is known as \emph{Petz recovery map or transpose map}.
Again, to ease notation, we drop the identities whenever they are obvious, for instance, for the expressions $\rho_{BC}^{\frac{1}{2}}((\rho_B^{-\frac{1}{2}}X \rho_B^{-\frac{1}{2}})\otimes \id_{C})\rho_{BC}^{\frac{1}{2}}$ we write just $\rho_{BC}^{\frac{1}{2}}\rho_B^{-\frac{1}{2}}X \rho_B^{-\frac{1}{2}}\rho_{BC}^{\frac{1}{2}}$, and the same for $\log\rho_{ABC}-(\log\rho_{AB})\otimes \id_{C}$, which we write just $\log\rho_{ABC}-\log\rho_{AB}$.

Observe that we can also recover a tripartite density operator from $\rho_{BC}$ through $\mathcal{P}_{B\to AB}(\cdot)$:
	\begin{align}
	\rho_{AB}^{\frac{1}{2}}\rho_B^{-\frac{1}{2}}\rho_{BC} \rho_B^{-\frac{1}{2}}\rho_{AB}^{\frac{1}{2}},
	\end{align}
	and by uniqueness, since the von Neumann entropy is operator-concave~\cite{low:34,ben:zyc:06}, they are the same. However,  it is not know whether, given a family of QMC that can be constructed in polynomial time via $\mathcal{P}^{(n)}_{B\to BC}(X)$, it is possible to build $\mathcal{P}^{(n)}_{B\to AB}(X)$ in polynomial-time. The next result states that solution to MECMP (Problem 2) and also QCMP (Problem 1), for 3-chains can be fully determined when a QMC belongs to the compatibility set, the proofs can be found in~\cite{dig:mer:mat:19}.

\begin{lemma}\em \label{lemma: QMC - comp}
	Given a 3-chain $\{\rho_{AB},\rho_{BC}\}$ compatible with a QMC, say $\rho_{ABC}$, then the solution of the maximum entropy estimator $\widetilde{\rho}_{ABC}$ is precisely  $\rho_{ABC}$. Moreover, the 3-chain $\{\rho_{AB},\rho_{BC}\}$ is compatible with a QMC in $\mathcal{B}\left(\mathcal{H}_{ABC}\right)$ iff $\tr_A(\rho_{AB})=\tr_C(\rho_{BC})$ and the operator $\Theta_{ABC}=\rho_{BC}^\frac{1}{2}\rho_{B}^{-\frac{1}{2}}\rho_{AB}^\frac{1}{2}$ is normal. Moreover, 
	if two marginals $\{\rho_{AB},\rho_{BC}\}$ are compatible with a QMC on $\liouville{ABC}$, say $\rho_{ABC}$, then the operator $\Theta_{ABC}$ is its square root.
\end{lemma}

We are now able to extend the above result from 3-chains to a much more general setting, namely to trees. From this point on, we make the following assumption.

\begin{assumption}\label{assump:1}\em Assume  the graph $\mathcal{G}_\mathcal{C}$ associated to a Maximum Entropy Compatible Marginal Problem $\mathcal{C}$ over $\mathbb{X}=\{X_1,\dots X_n\}$ is a tree, that is, $\mathcal{G}_{\mathcal{C}}$ is an acyclic connected graph over $\mathbb{X}$. 
\end{assumption}

By taking any node as a root of $\mathcal{G}_{\mathcal{C}}$, we construct an arborescence (or a directed tree). For the sake of readability, we introduce the following notation. We call a 
    {\em constructive ordering of $\mathcal{C}$} any total order compatible with the topological order of an arborescence of $\mathcal{G}_{\mathcal{C}}$. W.L.O.G we consider a constructive order of the form $X_1<\dots <X_n$ and denote by $\mathcal{G}_k$ the induced subgraph of $\mathcal{G}_{\mathcal{C}}$ containing all the nodes $V_k=\{X_1,\dots X_k\}$ for $k\in \{1\dots n\}$. We also denote by ${\mathcal{C}_k}$ the marginals in $\mathcal{C}$ containing nodes in $\{X_1,\dots X_k\}$ and by $Y_k$, for $k\geq 2$, the node in $V_{k-1}$ connected to $X_k$ in $\mathcal{G}_k$ (the adjacent node of $X_k$ in $\mathcal{G}_k$). Finally, we denote by $\overline{Y_k}$ the set $V_{k-1}\setminus \{Y_k\}$, which is non-empty for $k\geq 3$.

	The next result follows easily:
	
\begin{proposition}
	\label{prop:leaf tree}
\emph{	 If $\mathcal{G}_{\mathcal{C}}$ is a tree, than all the subgraphs $\mathcal{G}_k$ are trees, and moreover, $X_k$ is a leaf of $\mathcal{G}_k$.}
\end{proposition}

%\begin{definition}
%Let $\rho\in\liouville{\mathbb{X}}$ with $\mathbb{X}:=\{X_1,\dots X_n\}$ be an invertible density operator over and 
%$\mathcal{C}$ is a (non-singular) set of two-body marginals of $\rho$. We say that $\rho$ is {\em factorizable via Petz according to $\mathcal{C}$} if its square root $\Theta\in\liouville{\mathbb{X}}$, i.e., it is the d.o. s.t. $\rho=\Theta\Theta^\dagger=\Theta^\dagger\Theta$ which admits a decomposition
%		\begin{equation}\label{eq:4}
%		\Theta=\left[ \prod_{\mathcal{C}\backslash\{\rho_{X_{\ell}Y_{k}}\}}\rho_{X_iX_j}^{\frac{1}{2}}\left( \id_{X_i}\otimes\rho_{X_j}^{-\frac{1}{2}}\right)\otimes\id_{\overline{\{X_iX_j\}}}\right] \rho_{X_{\ell}Y_{k}}^{\frac{1}{2}}\otimes\id_{\overline{\{X_\ell X_k\}}}.
%		\end{equation}
%\end{definition}
%}

We now define a quantum Markov tree, which, as we shall see later on, generalizes the notion of Markov random field, when the underlying graph is a tree.

\begin{definition}
	\em
Let $\rho\in\liouville{\mathbb{X}}$ with $\mathbb{X}:=\{X_1,\dots X_n\}$ be an invertible density operator over and 
$\mathcal{C}$ is a (non-singular) set of two-body marginals of $\rho$. We say that $\rho$ is {\em quantum Markov tree (QMT)} or is {\em factorizable via Petz according to $\mathcal{C}$} if its square root is such that $\rho=\Theta\Theta^\dagger=\Theta^\dagger\Theta$ where $\Theta$ admits a decomposition, for some constructive order $X_1<\dots<X_n$, of the form
		\begin{equation}\label{eq:4}
		\Theta=\Delta_n\dots\Delta_3 (\rho_{X_{1}X_{2}}^{\frac{1}{2}}\otimes\id_{\overline{\{X_1 X_2\}}})
		\end{equation}
with $\Delta_k=\left(\rho_{X_kY_k}^{\frac{1}{2}}\left( \id_{X_i}\otimes\rho_{Y_k}^{-\frac{1}{2}}\right)\right)\otimes\id_{\overline{\{X_kY_k\}}}$, for all $k=3\dots n$.
\end{definition}

We note that for Eq.~\eqref{eq:4} to be well defined, it must be the case that $\mathcal{G}_\mathcal{C}$ is a tree, that is, that we are working under Assumption~\ref{assump:1}. It is relatively simple to extend the notion to acyclic graphs (which may not be connected). The following result will shed some light on the relationship between Markov random fields and QMTs.

	\begin{theorem}
		\label{thm: QMT- Petz fact - log}
		 \em Let $\rho\in\liouville{\mathbb{X}}$ be an invertible density operator over and 
$\mathcal{C}$ is a (nontrivial) set of two-body marginals s.t. $\mathcal{G}_{\mathcal{C}}$ is a spanning tree over $\mathbb{X}$, then there exists  
$\rho\in \mathrm{Comp}(\mathcal{C})$ factorizable via Petz according to $\mathcal{C}$ iff there exists  
$\rho\in \liouville{\mathbb{X}}$ such that, equivalently, one of the following two hold:
	\begin{itemize}
\item[i)] $\log\rho=\sum_{\mathcal{C}}\log\rho_{X_iX_j}-\sum_{i=1}^{n}(\mbox{deg}\left(X_i \right)-1 )\log\rho_{X_i}$;
\item[ii)] we have for some constructive ordering $X_1<\dots<X_n$:$$\forall k=2, \dots, n:\quad \rho_{k}=\pTr{\overline{V_k}}{\rho}\ \textnormal{is s.t. }\ I_{\rho_k}\left(X_k:\overline{Y_k}| Y_k \right)=0.$$
	\end{itemize}	

	\end{theorem}
	\vspace*{-0.2cm}

\begin{proof}
The proof follows by induction on $k$, that is, by adding one edge per node following a constructive ordering in $\mathcal{C}$. So we have that
\begin{equation}\label{eq:cconstord}%\small
\mathcal{C}=\left\lbrace \rho_{X_{k}Y_{k}}\in\mathcal{B}\left(\hilbert{X_kY_k} \right):\  Y_k\in\{X_1,\dots, X_{k-1}\};\ k=2,\dots,n\right\rbrace,
\end{equation}

\noindent
The proof follows by complete induction on $k$.\\[2mm]

\noindent (Basis $k=3$): The first chain occurs when the third node is added, that is, when $k=3$. Assume there exists $\rho_3\in\mbox{Comp}\left( \mathcal{C}_3\right)$ that is factorizable via Petz, i.e.:
\begin{equation}
\Theta_3=\rho_{3}^\frac{1}{2}=\rho_{X_3Y_3}^\frac{1}{2}\rho_{Y_3}^{-\frac{1}{2}}\rho_{\overline{Y_3}Y_3}^\frac{1}{2}=\rho_{X_3Y_3}^\frac{1}{2}\rho_{Y_3}^{-\frac{1}{2}}\rho_{X_1X_2}^\frac{1}{2}.
\end{equation}
Observe that we can use Lemma~\ref{lemma: QMC - comp} and so, $\Theta_3$ is exactly the operator described in the lemma, and since it is a square root, it is normal. Then, by Lemma~\ref{lemma: QMC - I(A:C|B)} we have the following equivalences: $\rho_3$ is a QMC iff 
$I_{\rho_{3}}\left(X_3:\overline{Y_3}| Y_3 \right)=0$ iff
$\log\rho_{3}=\log \rho_{X_{3}Y_{3}}+\log\rho_{X_{1}X_{2}}-\log\rho_{Y_{3}}$. The other direction follows immediately.\ \\[2mm]

\noindent
(Induction step $k\longrightarrow k+1$):\ \\[2mm]
\noindent
Complete induction hypothesis: $\forall j=3,\dots,k\ \exists\,\rho_j\in\mbox{Comp}\left( \mathcal{C}_j\right) $ factorizable via Petz according with $\mathcal{C}_{j}$ iff 
there exists $\rho_j\in \liouville{V_j}$ such that, equivalently, one of the following two hold:
\begin{itemize}
	\item $\log\rho_j=\sum_{\mathcal{C}_j}\log\rho_{X_iX_t}-\sum_{i=1}^{j}(\mbox{deg}_{\mathcal{G}_j}\left(X_i \right)-1 )\log\rho_{X_i}$;
	\item $I_{\rho_j}\left(X_j: \overline{Y_j}| Y_j \right)=0$.
\end{itemize}

\noindent
Induction step: Assume there $\exists\,\rho_{k+1}\in\mbox{Comp}\left( \mathcal{C}_{k+1}\right)$ factorizable via Petz according with $\mathcal{C}_{k+1}$, then, our goal is to show that the following holds for $\rho_{k+1}$:
\begin{itemize}
\item $\log\rho_{k+1}=\sum_{\mathcal{C}_{k+1}}\log\rho_{X_{i}X_{t}}-\sum_{i=1}^{{k+1}}(\mbox{deg}_{\mathcal{G}_k}\left(X_i \right)-1 )\log\rho_{X_i}$ \textnormal{ and },
\item $I_{\rho_{k+1}}\left(X_{k+1}:\overline{Y_{k+1}}| Y_{k+1} \right)=0$.
\end{itemize}
So, assume $\exists \rho_{k+1}\in\mbox{Comp}\left( \mathcal{C}_{k+1}\right) $ factorizable via Petz, i.e.:
\begin{equation}
\Theta_{k+1}=\rho_{{k+1}}^\frac{1}{2}=\Delta_{k+1}\Delta_{k}\dots\Delta_3\rho_{X_1X_2}^\frac{1}{2}\quad \textnormal{where}\quad \Delta_{i}:=\rho_{X_iY_i}^\frac{1}{2}\rho_{Y_i}^{-\frac{1}{2}}.
\end{equation}
Then:
\begin{equation*}
\begin{split}
\rho_{k+1}&=\Theta_{{k+1}}\Theta_{{k+1}}^\dagger\\ &=\Delta_{k+1}\Delta_{k}\dots\Delta_2\rho_{X_1Y_1}\Delta_2^\dagger\dots\Delta_{k}^\dagger\Delta_{k+1}\\
&=\rho_{X_{k+1}Y_{k+1}}^\frac{1}{2}\rho_{Y_{k+1}}^{-\frac{1}{2}}\rho_k\ \rho_{Y_{k+1}}^{-\frac{1}{2}}\rho_{X_{k+1}Y_{k+1}}^\frac{1}{2}\\[1.5ex]
&=\Theta_{{k+1}}^\dagger\Theta_{{k+1}}\\ &=\rho_{X_1X_2}^\frac{1}{2}\Delta_3^\dagger\dots\Delta_{k}^\dagger\Delta_{k+1}\Delta_{k+1}\Delta_{k}\dots\Delta_3\rho_{X_1X_2}^\frac{1}{2}\\&=\rho_{k}^\frac{1}{2}\rho_{Y_{k+1}}^{-\frac{1}{2}}\rho_{X_{k+1}Y_{k+1}}\ \rho_{Y_{k+1}}^{-\frac{1}{2}}\rho_{k}^\frac{1}{2}.
\end{split}
\end{equation*}
We can use Lemma~\ref{lemma: QMC - comp} on the set $\set{\rho_{X_{k+1}Y_{k+1}}, \rho_k}$ and conclude that $\rho_{k+1}$ is a QMC in the order $X_{k+1} - \overline{Y_{k+1}} - Y_{k+1}$. So, using Lemma~\ref{lemma: QMC - I(A:C|B)}, we have
$\rho_{k+1}$ is a QMC iff  $I_{\rho_{k+1}}\left(X_{k+1}:\overline{Y_{k+1}}| Y_{k+1} \right)=0$ iff 
\begin{equation*}
\begin{split}
\log\rho_{k+1}&= \log \rho_{X_{k+1}Y_{k+1}}+\log \rho_{\overline{Y_{k+1}}Y_{k+1}}-\log(\rho_{Y_{k+1}})\\
&\stackrel{I.H.}{=}\sum_{\mathcal{C}_{k+1}}\log\rho_{X_{i}X_{t}}-\sum_{i=1}^{{k+1}}(\mbox{deg}_{\mathcal{C}_i}\left(X_i \right)-1 )\log\rho_{X_i}.
\end{split}
\end{equation*}

The other direction is  straightforward. Just notice that $\tr_{X_{k+1}}(\rho_{k+1})=\rho_{k}$, and by induction hypothesis $\rho_{k}$ is compatible with $\mathcal{C}_k$, and so is $\rho_{k+1}$. Moreover, by construction of $\rho_{k+1}$ it is also compatible with $\mathcal{C}_{k+1}$.
\end{proof}

Note that the proof of the previous theorem does not depend on which constructive ordering one chooses. This fact follows from the fact that condition $i)$ is equivalent to condition $ii)$, and condition $i)$ does not assume any ordering. 

The reader conversant in Markov random fields will identify condition ii) as the quantum analogous of the {\em Local Markov Property} of a Markov random field - any variable $X_i$ is conditionally independent of the remaining nodes given its adjacent nodes:
$$X_i\indep \overline{\{X_i\}\cup \text{Ad} X_i}\mathbin{|}\text{ad}(X_i),$$ where $\text{Ad} X_i$ is the set of adjacent nodes to $X_i$. The notion of conditional independence is equivalently replaced by the conditional mutual information being null, that is
$$I(X_i:\overline{\{X_i\}\cup \text{Ad} X_i}\mathbin{|} \text{Ad} X_i)=0,$$
which, for the case of the tree $\mathcal{G}_k$ and for the node $X_k$ we have
$$I(X_k:\overline{Y_k}\mathbin{|}Y_k)=0.$$

The following results states that how to compute the solution Maximum Entropy Compatible Marginal Problem when $\mathcal{G}_\mathcal{C}$ is a tree and there exists $\rho\in\mbox{Comp}(\mathcal{C})$ that factorizes via Petz according to $\mathcal{C}$.

\begin{corollary}\label{cor1}\em  Let
$\rho\in\liouville{}$ factorizes via Petz according to $\mathcal{C}$ and $\mathcal{G}_\mathcal{C}$ a spanning tree. Then
\begin{equation}
\rho=\underset{\rho'\in\mbox{Comp}(\mathcal{C})}{\mbox{arg}\ \mbox{max}} S(\rho').
\end{equation}
\end{corollary}
\begin{proof}
	It follows that in the case $\rho\in\liouville{}$ factorizes via Petz according to $\mathcal{C}$ we have $\log\rho=\sum_{\mathcal{C}}\log\rho_{X_iX_j}-\sum_{i=1}^{n}(\mbox{deg}\left(X_i \right)-1 )\log\rho_{X_i}$, which saturates the subadditivity of the von Neumann entropy for every 3-chain $\overline{Y_k} - Y_k - X_k$, $k=3,\dots ,n$ in the spanning tree.
\end{proof}

We are now ready to state our main theorem, which gives a stronger characterization for the existence of a compatible density operator that is a QMT. Previously, we needed multivariate measurements to establish whether there exists a QMT in the given compatibility set. Herein, we show that it is enough to consider two-body measurements, which makes the procedure feasible in practice. The proof requires some technical lemmas that we placed in Appendix~\ref{ap3}. 

\begin{theorem}
			\label{thm: MQT efficient}\em 
		Let
		$\mathcal{C}:=\{\rho_{X_iX_j}\in\mathcal{B}\left( \hilbert{X_iX_j}\right) , i\neq j\in\{1,\dots,n\}\}$ be a set of admissible two-body marginals and such that the associate graph $\mathcal{G}_\mathcal{C}=\left(V, E\right)$ is a spanning tree.
		Then, there exists $\tilde\rho\in\liouville{}$ such that  $\tilde\rho\in\text{Comp}\left(\mathcal{C}\right)$ factorizable via Petz according to $\mathcal{C}$ iff
		\begin{equation}
		\label{eq:thm: MQT efficient}
		I_{\rho}\left(X_i:\text{ad}\,X_j|\,X_j\right)=0,\ \forall \rho_{X_{i}X_{j}}\in\mathcal{C}\ \textnormal{and}\quad \forall \text{ad}\,X_j,\, \text{ad}\,X_j\neq X_i,
		\end{equation}
		where $\text{ad}\,X_j$ indicates an adjacent node of $X_j$ in $\mathcal{G}_\mathcal{C}$, that is $\text{ad} X_i\in \text{Ad} X_i$. Moreover, $$\tilde\rho:=\mbox{arg}\underset{{}^{\rho\in\text{Comp}\left(\mathcal{C}\right)} }{\mbox{max}}S(\rho).$$
\end{theorem}

\begin{proof}
	As in the previous theorem, we assume a constructive ordering $X_1<\dots <X_n$ for $\mathcal{C}$ which will be used in the induction proof. Moreover, we can rewrite $\mathcal{C}$ using such order as in Eq.~\eqref{eq:cconstord}. Thus, the set of conditions in Eq.~\eqref{eq:thm: MQT efficient} are:
 	\begin{equation}
 I_\rho\left(X_k:\text{ad}\,Y_k|\,Y_k\right)=0,\ \forall\,\text{ad}\,Y_{k}\in V_{k-1},\ k=3,\dots,n.
 \end{equation}
$\mathbf{(\Rightarrow)}$ Using the previous theorem  we have that 
$$I_{\rho_{k}}(X_k:\overline{Y_k}|Y_k)=I_{\rho}(X_k:\overline{Y_k}|Y_k)=0.$$ Moreover, by Proposition~\ref{prop:leaf tree}, $X_k$ is leaf in $\mathcal{G}_k$ and it is only connected to $Y_k$. Finally, by applying the chain rule of  the quantum conditional mutual information (c.f. in Appendix~\ref{ap3} Eq.~\eqref{eq:QCMI-chain-rule}) and choosing the chain to start in a node adjacent to $X_k$, say $\text{ad} {X_k}$, it follows that $I_\rho\left(X_k:\text{ad}\,Y_k|\,Y_k\right)=0$.\\ [2mm] 
%WE ARE HERE. Reminding that $V_{k-1}=V_k\backslash\set{X_k}$
%\begin{equation}
%	I_\rho\left(X_{k}:V_k\backslash\{X_{k}, \mbox{ad}X_{k}\}|\,\mbox{ad}X_{k}\right)=0,\,  \Rightarrow\
%	I_\rho\left(X_k:\text{ad}\,Y_k|\,Y_k\right)=0,\ \forall\,\text{ad}\,Y_{k}\in V_{k-1},\ k=3,\dots,n.
%\end{equation}
\noindent
$\mathbf{(\Leftarrow)}$ 
The proof follows again by complete induction in the number of nodes $k$, following the assumed constructive ordering of $\mathcal{C}$. Again, the simplest tree where the equation has any meaning requires three nodes.\\[-2mm]

\noindent (Basis $k=3$): for this case the statement of this theorem coincides with ii) of Theorem~\ref{thm: QMT- Petz fact - log}, since $\text{ad}\,Y_3=\overline{Y_3}$.\ \\[2mm]
%
%It follows since by hypotesis:\\
%
%$\mbox{ad}Y_3\in V_3\backslash\{X_3,Y_3\}=\set{X_1,X_2,X_3}\backslash\{X_3,Y_3\}=\set{X_1,X_2}\backslash\{Y_3\}\ \Rightarrow$ $\set{X_1,X_2}\equiv\set{\mbox{ad}Y_3,Y_3}$.\\
\noindent
(Induction step $k\longrightarrow k+1$):\ \\[2mm]
\noindent
Complete induction hypothesis: We assume 
$$
I_\rho\left(X_k:\text{ad}\,Y_k|\,Y_k\right)=0,\ \forall\,\text{ad}\,Y_{k}\in V_{k-1},\ k=3,\dots,n,
$$
and so, by hypothesis, $\rho_\ell$ is factorizable via Petz according to $\mathcal{C}_\ell$, and so, by Theorem~\ref{thm: QMT- Petz fact - log}, we have 
\begin{equation}\label{eq:consIHe}
	I_{\rho_{\ell}}\left(X_\ell:\overline{Y_\ell}|Y_\ell \right)=0 \ \forall \ell=3,\dots k.
\end{equation}
 Induction step:
 We assume $I_{\rho}(X_{k+1}:\mbox{ad}Y_{k+1}|Y_{k+1})=0\ \forall\,\text{ad}\,Y_{k+1}\in V_{k}$
and our goal is to show that there exists $\rho_{k+1}$ factorizable via Petz according to $\mathcal{C}_{k+1}$ such that its partial traces hold
$$
 I_{\rho_{k+1}}\left(X_{k+1}:\overline{Y_{k+1}}|Y_{k+1} \right)=0.
$$

Observe that, by definition, $Y_{k+1}\in V_k$, let $m_{k+1}$ be some step in which $Y_{k+1}$ was connected to some node (note that it might connect to some node in many steps). Clearly, we have $3\leq m_{k+1}\leq k$. We consider two cases, depending on the degree of $Y_{k+1}$ in $\mathcal{G}_{k}$.

\ \\[2mm]
Case 1) deg$Y_{k+1}$=1, then by construction, it must be that $Y_{k+1}=X_{m_k}$ and by Eq.\eqref{eq:consIHe} we have that for $\rho_{m_k}$ its partial traces hold $$I_{\rho_{m_k}}\left(X_{m_k}:\overline{Y_{m_k}}|Y_{m_k} \right)=0.$$
	By Lemma~\ref{lemma:relax} (in Appendix~\ref{ap3}) since 
	$$V_k\setminus\{X_{m_k},Y_{m_k}\}\supseteq \overline{Y_{m_k}}=V_{m_k}\setminus\{X_{m_k},Y_{m_k}\},$$ we also have for $\rho_k$ that 
	$$I_{\rho}\left(X_{m_k}:V_k\setminus\{X_{m_k},Y_{m_k}\}|Y_{m_k} \right)=I_{\rho}\left(Y_{k+1}:V_k\backslash\{Y_{k+1},\mbox{ad}Y_{k+1}\}|\mbox{ad}Y_{k+1} \right)=0,$$
	where the last equality is obtained by noticing that $X_{m_k}=Y_{k+1}$ and $Y_{m_k}=\mbox{ad}Y_{k+1}$. Recall that we have,
	%Then, since by assumption %$I_{\rho_{k+1}}(X_{k+1}:\mbox{ad}Y_{k+1}|Y_{k+1})=0$, by chain rule of conditional mutual information we have
	$$
	I_\rho\left(X_{k+1}:\mbox{ad}\,Y_{k+1}|Y_{k+1} \right)=0.
$$
	Moreover, the set $\{V_k\backslash\{Y_{k+1},\mbox{ad}Y_{k+1}\},\mbox{ad}\,Y_{k+1},Y_{k+1},X_{k+1}\}$, forms the chain
	\begin{equation*}
	V_k\backslash\{Y_{k+1},\mbox{ad}Y_{k+1}\}-\mbox{ad}\,Y_{k+1}-Y_{k+1}-X_{k+1}.
	\end{equation*}
	Then, by using Lemma~\ref{lemma4} (a) (in Appendix~\ref{ap3}), there exists a density operator $\rho_{k+1}\in\liouville{V_{k+1}}$ such that its partial traces fulfill
	\begin{equation*}
	\qI{X_{k+1}:V_k\backslash\{Y_{k+1}\}|Y_{k+1}}=I_{\rho_{k+1}}(X_{k+1}:\overline{Y_{k+1}}|Y_{k+1})=0.
	\end{equation*}
	Moreover, by construction of this $\rho_{k+1}$ in Lemma~\ref{lemma4} (a) (in Appendix~\ref{ap3}) we have	
	$\pTr{X_{k+1}}{\rho_{k+1}}=\rho_k$, and so $\rho_{k+1}$ is s.t.:
	\begin{equation*}
	\qI{X_{i}:V_{i}\backslash\{X_{i},Y_{i}\}|Y_{i}}=0\quad \forall i:\,2\leq i\leq k+1.
	\end{equation*}
Case 2) deg$Y_{k+1}>1$, then $\mathcal{G}_{k+1}$ can be seen as a star centred in $Y_{k+1}$, with as many branches, as many as adjacent nodes $(\mbox{ad}Y_{k+1})_i$  in $\mathcal{G}_{k+1}$, whose number is precisely the degree $r_k$ of $Y_{k+1}$ in $\mathcal{G}_k$, plus the new added node $X_{k+1}$ (c.f. Fig.~\ref{fig:main thm:proof:stara}).
\begin{figure}[htp]
			
		\begin{displaymath}
		\xymatrix{ & Y_{k+1} \ar@{-}[dl]\ar@{-}[d]\ar@{-}[dr] \ar@{-}[drr]& \\
			X_{k+1} &(\mbox{ad}\,Y_{k+1})_1 \ar@{-}[d]& \{\dots\}&(\mbox{ad}\,Y_{k+1})_{r_{k}} \ar@{-}[d]\\
		&\mathcal{G}_1& & \mathcal{G}_{r_k}}
		\end{displaymath}
			\caption{The associate graph $\mathcal{G}_{k+1}$: can be seen as a star centred in $Y_{k+1}$, where every branch is an adjacent of $Y_{k+1}$ in $V_k$, plus the link to $X_{k+1}$. $\mathcal{G}_i$ indicates the rest of the graph (a tree) that is connected to the i-th adjacent $(\mbox{ad} Y_{k+1})_i$. The number of adjacent nodes to $Y_{k+1}$ in $\mathcal{G}_{k+1}$ is $r_k+1$ by adding $X_{k+1}$ to other $r_k$ nodes in $\mathcal{G}_{k}.$  
			}\label{fig:main thm:proof:stara}
	\end{figure}
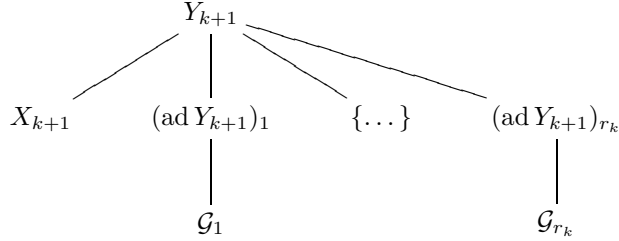
	To prove the thesis we must find $\rho_{k+1}$  such that, if
\begin{equation}
    I_\rho\left(X_{k+1}:(\mbox{ad}Y_{k+1})_i|Y_{k+1}\right)=0\  \forall i=1\dots r_k
\end{equation} then, accordingly to Theorem~\ref{thm: QMT- Petz fact - log}, it is enough to show:
\begin{equation}
\label{eq:proofthm3:k+1stepthm2}
    I_{\rho_{k+1}}\left(X_{k+1}:\overline{Y_{k+1}}|Y_{k+1}\right)=0.
\end{equation}
Moreover, by inductions hypothesis we know that 
	\begin{equation}
	\label{eq:proofthm3:kstep}
I_\rho\left(X_\ell:\mbox{ad}Y_\ell|Y_\ell\right)=0\ \forall \mbox{ad}Y_\ell\in V_{k},\ \ell=3, \dots k.
	\end{equation} 
and again, by Theorem~\ref{thm: QMT- Petz fact - log}, we must have:
\begin{equation}\label{eq:cih}
  I_{\rho_\ell}\left(X_\ell:\overline{Y_\ell}|Y_\ell\right)=0\ \forall \ell=3, \dots k. 
\end{equation}
We proceed to show Eq.~\eqref{eq:proofthm3:k+1stepthm2} by using
Corollary~\ref{corollary: (lemma8):n-parties:chain e star} (in Appendix~\ref{ap3}). Indeed, this results guarantees that the star $$\{X_{k+1}, Y_{k+1}, (\mbox{ad}Y_{k+1})_1\cup \mathcal{G}_1, \dots, (\mbox{ad}Y_{k+1})_{r_k}\cup \mathcal{G}_{r_k}\}$$ factorizes via Petz according to $$\{X_{k+1}Y_{k+1}, Y_{k+1}(\mbox{ad}{Y_{k+1}})_1\cup\mathcal{G}_1, \dots, 
Y_{k+1}(\mbox{ad}{Y_{k+1}})_{r_k}\cup\mathcal{G}_{r_k}\}$$ iff
	 \begin{align}
	 \label{eq:mainthm:cond:star}
 &\qI{X_{k+1}:\left(\mbox{ad}Y_{k+1}\right)_i\cup\mathcal{G}_i\ |\ Y_{k+1}}=0,\quad \forall i\in 1,\dots, r_k;\\[1.5ex]
 	 \label{eq:mainthm:cond:star2}
  &\qI{\left(\mbox{ad}Y_{k+1}\right)_i\cup\mathcal{G}_i:\left(\mbox{ad}Y_{k+1}\right)_j\cup\mathcal{G}_j\ |\ Y_{k+1}}=0,\quad \forall i\neq j\in 1\dots r_k.
  \end{align}
 From which, by Theorem~\ref{thm: QMT- Petz fact - log}, we get the goal, stated in Eq.~\eqref{eq:proofthm3:k+1stepthm2}.
	 
The conditions in Eq.~\eqref{eq:mainthm:cond:star2} come from the complete induction hypothesis Eq.~\eqref{eq:cih}. 
%Observe given a set of two body marginals  $\mathcal{C}$, all the constructive ordering of the nodes for constructing the associate graph are equivalent. Then, once constructed $\mathcal{G}_i$, we can continue the construction as
%\begin{equation}
%\mathcal{G}_i - (\mbox{ad}Y_{k+1})_i - Y_{k+1} - %(\mbox{ad}Y_{k+1})_j,
%\end{equation}
%from we get the condition:
%bla
%So using Prop.~\ref{prop: QMT- Petz fact - log}, it follows the state $\rho_k$ is a quantum Markov network, then it verifies all the conditional independence conditions on the sub-tripartite chains.\\
On the other hand, the conditions stated in Eq.~\eqref{eq:mainthm:cond:star}, come from observing that, for every $(\mbox{ad}Y_{k+1})_i$, there is a chain
\begin{equation}
X_{k+1} - Y_{k+1} - (\mbox{ad}Y_{k+1})_i - \mathcal{G}_i,
\end{equation}
for which we already have the conditions:
\begin{align}
\qI{X_{k+1}: (\mbox{ad}Y_{k+1})_i | Y_{k+1}}=0,\label{eq:chain1}\\[1ex]
\qI{Y_{k+1}: \mathcal{G}_i | (\mbox{ad}Y_{k+1})_i}=0.\label{eq:chain2}
\end{align}
Eq.~\eqref{eq:chain1} follows from induction hypothesis Eq.~\eqref{eq:proofthm3:kstep}. Moreover,  Eq.~\eqref{eq:chain2} follows from the fact that, by hypothesis, $\rho_k$ is a QMT, and so  $$Y_{k+1} - (\mbox{ad}Y_{k+1})_i -\mathcal{G}_i$$ is a quantum Markov chain. So, by using Lemma~\ref{lemma4} (a) (in Appendix~\ref{ap3}), we get the desired condition $$\qI{X_{k+1}:\left(\mbox{ad}Y_{k+1}\right)_i\cup\mathcal{G}_i\ |\ Y_{k+1}}=0.$$
Since the argument holds for all the adjacent nodes $\left(\mbox{ad}Y_{k+1}\right)_i$, we derive the whole set of conditions \eqref{eq:mainthm:cond:star}, which ends the proof for case 2)\ \\[2mm]

Finally, the fact that the obtained state maximizes the von Neumann entropy with the provided marginals, comes for free from Corollary~\ref{cor1}.
\end{proof}

We are now able to show that for QMTs, the MECM problem is in P and that there is a polynomial quantum circuit that constructs the Maximum entropy compatible density operator. Moreover, we also show that it is possible to extend the Chow-Liu algorithm efficiently for quantum Markov networks. To derive these results, we need first to compute the number of 3-chains in a graph with $n$ nodes.

\begin{lemma}\label{lemma:nchain}\em 
The number of 3-chains $\#c$ in a tree with $n\geq 2$ vertices is $n-2\leq\#c\leq \frac{1}{2}(n-1)(n-2)$. Moreover, the number of 3-chains for any graph is upper-bounded by $\frac{1}{2}n(n-1)(n-2)$, and it reaches the bound for a complete graph of $n$ nodes.   
\end{lemma}
\begin{proof}
We make the proof by counting, for each node $X_i$, how many 3-chains $X_j-X_i-X_k$ can be formed, and summing all of them afterwards.

For a spanning tree, the lower bound is the number of 3-chains in a $n$-chain (all nodes have degree 2, with exception of the root and the leaf). 
In this case, every node is the central node of only one 3-chain, aside for the root and the leaf; thus, $\#c=n-2$. 
The upper bound is derived by counting the number of 3-chains in a $n$-star (there is a root and all the remaining nodes are leaves).
The root, say $Y$, has $\mbox{deg}Y=n-1$, and the remaining nodes (enumerate them as $X_1,\dots X_{n-1}$), have degree one. In this case, consider the first edge $X_1Y$, it can be linked through Y to more n-2 nodes, which also gives the number of 3-chains it can be part of. 
The next edge $X_2Y$, it can be connected through $Y$ to $n-3$ nodes to form $n-3$ different chains  (the chain $X_2-Y-X_1$ is the same as $X_1-Y-X_2$, which has been already counted for). It is now clear that the number of 3-chains in an $n$-star is
\begin{equation}
\label{eq:prop3:3chain:nhub}
\#c_i=\sum_{k=2}^n(n-k)=\sum_{k=1}^{n-2} k=\frac{1}{2}(n-1)(n-2).
\end{equation}

The number of chains in a $n$-star is also the number of 3-chains that a node contributes in a complete graph. Then, to obtain the number of 3-chains in a complete graph it is enough to multiply Eq.~\eqref{eq:prop3:3chain:nhub} by the number of nodes, and so  $\#c=n\#c_i=\frac{1}{2}n(n-1)(n-2)$.

Another way of obtaining this value consists in using well-known formulas from combinatorial calculus, and observing that the number of 3-chains in a complete graph of $n$ vertices is the number of \emph{simple dispositions}, i.e. the number of ordered sequences of length 3 without repetitions in a set of $n$ elements, divided by two. The factor 2 comes from the symmetry of the 3-chains, that is $A-B-C$ is the same 3-chain as $C-B-A$. Then, once again, $$\#c=\frac{1}{2}\frac{n!}{(n-3)!}=\frac{1}{2}n(n-1)(n-2)$$
\end{proof}

We are now able to establish a sufficient condition for the MECMP problem to be in $P$.

	\begin{theorem}\em
		\label{thm: comp with QMN}
		 The Maximum Entropy Compatible Marginal Problem for $\mathcal{C}$ is in P when 
		 \begin{enumerate}
		 	\item $\mathcal{G}_C$ is a spanning tree
		 	\item $\rho_{ijk}$ is a QMC constructed in polynomial-time (with respect with the number of nodes $n$) where $\rho_{i,j},\rho_{j,k}\in \mathcal{C}$ and $i<j<k$ for some given constructive order of $\mathcal{G}_C$.
		 \end{enumerate}		  
		 Moreover, there exists a quantum polynomial circuit that constructs the maximum entropy compatible tree.
	\end{theorem}

\begin{proof}
	From Theorem~\ref{thm: MQT efficient}, the density operator that maximizes the Entropy is a QMT. Moreover, we can compute its entropy in polynomial time, by considering the  constructive ordering of point 2. Indeed, from Theorem~\ref{thm: QMT- Petz fact - log}~(i), when $\rho$ is a QMT we have that
	$$S(\rho)=\sum_{\mathcal{C}}S(\rho_{X_iX_j})-\sum_{i=1}^{n}(\mbox{deg}\left(X_i \right)-1 )S(\rho_{X_i}).$$
	
Moreover, since each $\rho_{X_iX_j}$ belongs to a QMC constructed in polynomial time, we can compute a matrix representation of the density operator of the QMC in polynomial-time as well. Recall in Definition~\ref{def:PTQMC}, that the Hilbert space of a polynomial-time QMC is fixed, and does not depend on the complexity parameter, that is, as usual, the dimension of the Hilbert space associated to each node is fixed (regarding) the complexity parameter $n$ (the number of nodes).

Moreover, given the constructive order, we are also able to make a quantum circuit to construct the maximum entropy compatible tree by constructing the first Markov chain $\rho_{X_1,X_2,X_3}$ and then applying the circuits for the recovery maps $\mathcal{R}$ of the remaining nodes.\\
\begin{figure}[htp]
\begin{center}
\hspace{1.3cm}\Qcircuit @C=1em @R=.7em @!R{
\lstick{\id_{ X_1}/d_{X_1} }& \gate{\Phi_{X_1}} & \multigate{1}{\mathcal{R}_{X_1X_2}} & \multigate{2}{\mathcal{I}_{\overline{Y_3}}\otimes\mathcal{R}_{Y_3X_3}} & \qw \\
\lstick{\id_{X_2}/d_{X_2} }& \qw & \ghost{\mathcal{R}_{X_1X_2}} & \ghost{\mathcal{I}_{\overline{Y_3}}\otimes\mathcal{R}_{Y_3X_3}} & \qw\\
\lstick{\id_{X_3}/d_{X_3} }& \qw & \qw & \ghost{\mathcal{I}_{\overline{Y_3}}\otimes\mathcal{R}_{Y_3X_3}} & \qw \\
\lstick{\cdots}& \qw & \qw & \qw & \qw \\
\lstick{\id_{X_n}/d_{X_n}} & \qw & \qw & \qw & \qw
} %\hspace{0.1cm}
\Qcircuit @C=1em @R=2.25em {\rstick{\cdots}\\
\rstick{\cdots}\\
\rstick{\cdots}\\
\rstick{\cdots}\\
\rstick{\cdots}
}
\hspace{0.7cm}
\Qcircuit @C=1em @R=1.35em {
& \multigate{4}{\mathcal{I}_{\overline{Y_n}}\otimes\mathcal{R}_{Y_nX_n}} & \qw \\
& \ghost{\mathcal{\mathcal{I}_{\overline{Y_n}}\otimes\mathcal{R}_{Y_nX_n}}} & \qw \\
& \ghost{\mathcal{\mathcal{I}_{\overline{Y_n}}\otimes\mathcal{R}_{Y_nX_n}}} & \qw \\
& \ghost{\mathcal{\mathcal{I}_{\overline{Y_n}}\otimes\mathcal{R}_{Y_nX_n}}} & \qw \\
& \ghost{\mathcal{\mathcal{I}_{\overline{Y_n}}\otimes\mathcal{R}_{Y_nX_n}}} & \qw 
}
\end{center}
\caption{Quantum circuit that outputs the optimal QMT. Note that the $k$-th block $\mathcal{I}_{\overline{Y_k}}\otimes\mathcal{R}_{Y_kX_k}$ operates only over two components: $X_k$ and $Y_k$, for all $k=3\dots n$.}\label{fig2}
\end{figure}
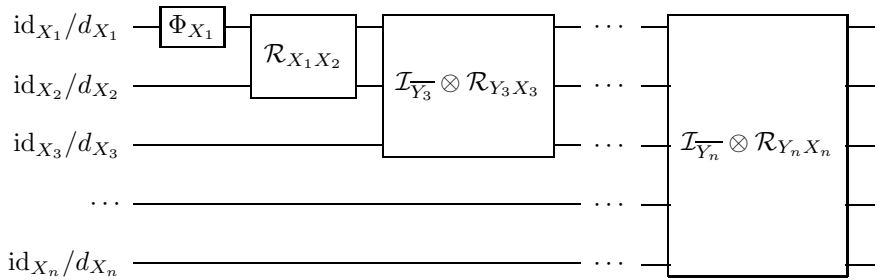
\ \\[-8mm]
\end{proof}
Two-body marginals for which all 3-chains form a QMC have another interesting property. It is possible to find the QMT closest, with regards to the quantum relative entropy (the generalization of the Kullback-Leibler divergence~\cite{cov:tho:12}), to the unknown  density operator. Note that, the number of spanning trees over a complete graph is given by Cayley’s formula~\cite{cay:89}, $n^{n-2}$
which is exponential on $n$. To extract the closest QMT, we need to construct a weighted graph (where the nodes are each component of the density operator), and the edges are weighted with the von Neumann mutual information between every two components. The optimal spanning tree (which can be found using the polynomial-time algorithm by Prim or Kruskal) gives the support to a QMT. Moreover, this QMT will be the one closest to the unknown state. When the density operators are diagonal, that is, describe a probability distribution; this algorithm coincides with the well-established Chow-Liu algorithm.

\begin{theorem}
\label{thm:chowLiu}\emph{
If the set of two body marginals $\mathcal{C}$ is s.t. every 3-chain is compatible with a QMC then  every subtree is a QMT. A QMT that minimizes the quantum relative entropy with respect to the (unknown) given quantum state, is the maximum weighted tree $\mathcal{G}_{T_C}$ where the weight of each edge is given by the quantum mutual information. Such tree can be obtained efficiently using the (generalized) Chow-Liu learning algorithm~\cite{cho:liu:68}.
}
\end{theorem}
\begin{proof}
The proof consists in applying Theorem~\ref{thm: MQT efficient} to the main result of section 6 in the paper \cite{dig:mer:mat:19}, that we are going to briefly recall.\\

Let $\rho_\mathbb{X}$ be the unknown quantum state that describes best the quantum system for which the bipartite marginals are known (for instance, they have been measured and collected in $\mathcal{C}$). Moreover, let $\widetilde{\rho}_{\mathcal{C}_T}$ be the maximum von Neumann entropy d.o. compatible with a subset ${\mathcal{C}}_T\subseteq\mathcal{C}$ s.t. $\mathcal{G}_{\mathcal{C}_T}$ is a tree. We refer to $\widetilde{\rho}_{\mathcal{C}_T}$ as quantum tree.
Their relative entropy can be written as
\begin{equation}
\label{eq:minRelativeEntropy}
S\left(\rho_{\mathbb{X}}||\widetilde{\rho}_{\mathcal{C}_T}\right)= -S\left(\rho_{\mathbb{X}}\right) +\tr\left(\rho_{\mathbb{X}}\log\widetilde{\rho}_{\mathcal{C}_{{T}}} \right)=S\left(\widetilde{\rho}_{\mathcal{C}_{{T}}}\right)-S\left(\rho_{\mathbb{X}}\right),
\end{equation}
where we have used condition (i) in Theorem~\ref{thm: QMT- Petz fact - log} on $\log\widetilde\rho_{\mathcal{C}_T}$.

So the optimal maximum entropy estimator $\widetilde\rho$ is computed over the subtree with minimal von Neumann entropy:
\begin{equation}
\label{eq:65}
\widetilde{\rho}=\underset{\mathcal{C_T}\subseteq\mathcal{C}}{\mbox{\textnormal{argmin}}}\,\underset{\rho\in{\small\mbox{Comp}(\mathcal{C}_T)}}{\mbox{\textnormal{max}}}\,S\left(\rho\right).
\end{equation}
Since, the number of possible spanning trees is $n^{n-2}$~\cite{cay:89},  we can not choose the best fitting tree efficiently, in general. However, in the case at hand, we can manipulate Eq.~\eqref{eq:minRelativeEntropy} and derive subcases for which the computation can be performed efficiently.

Observe that
$$\sum_{\mathcal{C}_T}S\left( \rho_{X_iX_j}\right)-\sum_{i=1}^{n}\left( \textnormal{deg}X_i-1\right) S\left( \rho_{X_i}\right)=-\sum_{\mathcal{C}_T}I_\rho(X_i,X_j)+\sum_{i=1}^{n}S\left( \rho_{X_i}\right),
$$
and set
\begin{equation}
\label{eq:deltaS}
\Delta S(\widetilde{\rho}_{\mathcal{C}_{\mathcal{T}}}):=\sum_{\mathcal{C}_T}S\left( \rho_{X_iX_j}\right)-\sum_{i=1}^{n}\left( \textnormal{deg}X_i-1\right) S\left( \rho_{X_i}\right) -S(\widetilde{\rho}_{\mathcal{C}_{\mathcal{T}}}),
\end{equation}
which is always non-nengative. By adding and subtracting the term 
$$\sum_{\mathcal{C}_T}S\left( \rho_{X_iX_j}\right)+\sum_{i=1}^{n}\left( \textnormal{deg}X_i-1\right) S\left( \rho_{X_i}\right)$$ to Eq.~\eqref{eq:minRelativeEntropy},
it assumes the form
\begin{equation}
\label{eq:minRelativeEntropy:deltaS}
S\left(\rho_{\mathbb{X}}||\widetilde{\rho}_{\mathcal{C}_{{T}}}\right)=-\sum_{\mathcal{C}_T}I_\rho(X_i,X_j)-\Delta S(\widetilde{\rho}_{\mathcal{C}_{\mathcal{T}}})+\sum_{i=1}^{n}S\left( \rho_{X_i}\right) -S\left(\rho_{\mathbb{X}} \right).
\end{equation}

%or quantum Markov trees we know that $\Delta S(\widetilde{\rho}_{\mathcal{C}_{\mathcal{T}}})=0$, the $\mathcal{C}_T$ that minimizes Eq.~\eqref{eq:minRelativeEntropy:deltaS}, as we are going to see, we know how to minimize the quantity in Eq.~\eqref{eq:minRelativeEntropy:deltaS} efficiently. \vi{there is a problem with the sign of $\Delta$ here and Eq.~(25).}
By using condition (i) of Theorem~\ref{thm: QMT- Petz fact - log} we can replace the log term of $S\left(\widetilde\rho_{\mathcal{C}_T}\right)$ of Eq.~\eqref{eq:deltaS} and thus, for a QMT, $\Delta S (\widetilde{\rho}_{\mathcal{C}})=0$. Moreover, we also have the converse, that is, $\Delta S (\widetilde{\rho}_{\mathcal{C}})=0$ holds only for QMTs. The latter result can be derived by observing that 
\begin{equation}
\Delta S(\widetilde{\rho}_{\mathcal{C}_{\mathcal{T}}}) =\sum_{i=1}^{n-2}I_{\rho}(X_{l_i}: V_i\backslash\{X_{l_i},\mbox{ad} X_{l_i}\}|\mbox{ad} X_{l_i}),
\end{equation} 
which, by positivity of quantum conditional mutual information, must be 0 iff all the terms in the sum are 0. Then, by Theorem~\ref{thm: MQT efficient}, we have $\Delta S (\widetilde{\rho}_{\mathcal{C}})=0$ iff all the 3-chains in $\mathcal{C}_T$ are QMC.

So, when the provided set of marginals $\mathcal{C}$ is s.t. every 3-chain is compatible with a QMC, $\Delta S (\widetilde{\rho}_{\mathcal{C}})=0$ in Eq.~\eqref{eq:minRelativeEntropy:deltaS}. Therefore, the best tree is the one that maximizes the term
\begin{equation}
\sum_{\mathcal{C}_T}I_\rho(X_i,X_j),
\end{equation}
i.e. the maximum weighted spanning sub-tree, where the weights are given by the  mutual information between every couple of linked nodes.

This problem is efficiently solved for classical graphs by the Chow-Liu algorithm, which  we have here generalized to quantum states, be replacing the Shannon entropy with the von Neumann entropy.
\end{proof}
The general case of efficiently finding the optimal spanning tree which gives the support to a quantum tree remains open. Minimizing the general form of Eq.~\eqref{eq:minRelativeEntropy:deltaS} would require the maximization of the quantity $\sum_{\mathcal{C}_T}I_\rho(X_i,X_j)+\Delta S(\widetilde{\rho}_{\mathcal{C}_{\mathcal{T}}})$ with $S(\widetilde{\rho}_{\mathcal{C}_{\mathcal{T}}})\!>\!0$. Already in the tripartite scenario, it is evident that the maximum weighted tree is not necessarily solution to the problem.\\

For sake of completeness, we present the quantum version of the Chow-Liu algorithm in pseudo-code below.
\begin{algorithm} Quantum Chow-Liu algorithm\\

\noindent
\textbf{Input:} $\mathcal{C}$ two-body marginals over $\mathbb{X}=\{X_1,\dots X_n\}$ of $\rho$, such that all 3-chains $X_i-X_j-X_k$, with $i,j,k\in $, form a polynomial-time QMC, with respect to $n$. \\
\noindent
\textbf{Output:} a QMT $\rho_T$ such that $S(\rho||\rho_T)$ is minimal.
\begin{enumerate}
    \item Construct the complete (undirected) graph over $\mathbb{X}$ where each edge $(i,j)$ is weighted with $S(\rho_{X_i}:\rho_{X_j})$;
    \item Extract the maximum spanning tree $T$ for this graph using, for instance, Prim algorithm;
    \item Use the quantum circuit Figure~\ref{fig2} (in the proof of Theorem~\ref{thm: comp with QMN}) to construct $\rho_T$ where we consider any constructive order for $T$.
\end{enumerate}

\end{algorithm}

\section{Conclusions and open problems}
We have showed that  comparing the entropies of 3-chains (the simplest non-trivial scenario) is QSZK-complete~\cite{wat:08,ber:vaz:97,wat:02}. Although this result hints that finding the maximum entropy compatible state should be not feasible, it does not answer the problem of finding a lower-bound for the Maximum Entropy Compatible problem.

We addressed this problem partially, by showing a subclass where it is in P. Concretely, we addressed the problem for the case where there is a compatible QMT (quantum Markov tree) whose 3-chains are polynomial-time quantum Markov chains. With these results, we were able to extend the Chow-Liu algorithm~\cite{cho:liu:68} for quantum states whose 3-subchains are quantum Markov chains. Moreover, we showed that, in the case at hand, the maximum entropy quantum state could be constructed by a polynomial-time quantum circuit.  

Understanding other classes of quantum states for which this problem is tractable (at least in quantum polynomial time) would be a relevant problem. We note that the Chow-Liu algorithm is a standard machine learning algorithm to learn the fittest Bayesian network tree that describes data. We also expect that the proposed extension for quantum states sheds some light on the power of quantum machine learning.

\bibliographystyle{unsrt}
\bibliography{bib1}

\appendix

\section{Lemmas for Theorem~\ref{thm1}}\label{ap1}
To perform the proof of Theorem~\ref{thm1} we need the following lemmas.

\begin{lemma}
\label{thm:polarization} \em \textbf{(Polarization lemma, Theorem 5 at
\cite{wat:02})} Let $\alpha$ and $\beta$ satisfy $0 \le \alpha <
\beta^2 \le 1$. Then there is a deterministic polynomial-time
procedure that, on input $(Q_0, Q_1, 1^n)$ where $Q_0$ and $Q_1$ are
quantum circuits, outputs descriptions of quantum circuits $(R_0,
R_1)$ (each having size polynomial in $n$ and in the size of $Q_0$
and $Q_1$) such that

\begin{eqnarray*}
\trnorm{\rho_0 - \rho_1} \le \alpha & \Rightarrow &
\trnorm{\mu_0 - \mu_1} \le 2^{-n}, \\
\trnorm{\rho_0 - \rho_1} \ge \beta & \Rightarrow &
\trnorm{\mu_0 - \mu_1} \ge 1-2^{-n}.
\end{eqnarray*}
\end{lemma}

The proof of the following lemmas can be found in \cite{nie:chu:12}.

\begin{lemma}
\label{Lemma:joint-entropy-theorem}\em \textbf{(Joint entropy theorem)}~\cite{nie:chu:12}
Suppose $p_i$ are probabilities, $\ket{i}$ are orthogonal states for
a system $A$, and $\rho_i$ is any set of density operators for
another system B. Then
$$S \left( \sum_i p_i \ketbra{i}{i} \tensor \rho_i \right) = H(p_i)+ \sum_i p_i S(\rho_i).$$
\end{lemma}

\begin{lemma}
\label{Lemma:Fannes}\em \textbf{(Fannes' inequality)}~\cite{nie:chu:12} Suppose $\rho$ and
$\sigma$ are density matrices over a Hilbert space of dimension $d$.
Suppose further that the trace distance between them satisfies $t =
\trnorm{\rho - \sigma} \le 1/e$. Then
$$|S(\rho) - S(\sigma)| \le t( \ln d - \ln t ).$$
\end{lemma}

\begin{lemma}
\label{lem:ANTV}\em \textbf{(Lemma 3.2 at \cite{amb:nay:tas:vaz:02})} Let $\rho_0$
and $\rho_1$ be two density matrices, and let $\rho = \half(\rho_0 +
\rho_1)$. If there exists is a measurement with outcome $0$ or $1$
such that making the measurement on $\rho_b$ yields the bit $b$ with
probability at least $p$, then
$$S(\rho) \ge \half[S(\rho_0) + S(\rho_1)] + (1-H(p)).$$
\end{lemma}

In particular, by choosing the right observable we have

\begin{lemma} \em\textbf{(Lemma B.6 at \cite{aro:tas:07})}
\label{lem:ANTV-with-trace-norm} Let $\rho_0$ and $\rho_1$ be two
density matrices, and let $\rho = \half(\rho_0 + \rho_1)$. Then
$$S(\rho) \ge \half[S(\rho_0) + S(\rho_1)] + (1-H(\half + \frac{\trn{\rho_0 - \rho_1}}{2})).$$
\end{lemma}

\section{Proof of the central Lemma\ref{lemma: QMC - I(A:C|B)}}
\label{appendix:proof:lemma6}
%\begin{lemma}\label{lemma: QMC - I(A:C|B)}\em Let $\rho_{ABC}$ be an invertible density operator. The following four assertions are equivalent:
%	\begin{enumerate}
%		\item $\rho_{ABC}$ is a QMC over the chain $A-B-C$.
%		\item $ I_{\rho}(A:C|B)=0$, where $I_{\rho}(A:C|B):=S(\rho_{AB})+S(\rho_{BC})-S(\rho_{B})-S(\rho_{ABC})$.
%		\item $\mathcal{P}_{B\to BC}(X):=\rho_{BC}^{\frac{1}{2}}((\rho_B^{-\frac{1}{2}}X \rho_B^{-\frac{1}{2}})\otimes \id_{C})\rho_{BC}^{\frac{1}{2}}, \text{ is a CPTP map for any } X\in \mathcal{L}(\mathcal{H}_{B})$ and preserves the partial trace $\rho_{AB}$.
%		\item $\log\rho_{ABC}-(\log\rho_{AB})\otimes \id_{C}=\id_{A}\otimes (\log\rho_{BC})-\id_{A}\otimes(\log\rho_B)\otimes \id_{C}$.
%	\end{enumerate}
%\end{lemma}
The proof of the lemma comes straightforwardly from the following definitions and previously established theorems.
\begin{definition}\emph{
		Let $\hilbert{}$ be a finite dimensional Hilbert space and $\rho_i\in\liouville{}$, i=1,2, density operators. Their \emph{relative entropy} is defined as:
		\begin{equation}
		S\left( \rho_1\|\rho_2\right):=\begin{cases}
		\mbox{Tr}\rho_1\left(\log\rho_1-\log\rho_2\right) & \text{if}\ \mbox{supp}\left(\rho_1\right)\subseteq \mbox{supp}\left(\rho_2\right)\\
		+\infty& \text{otherwise}.
		\end{cases} 
		\end{equation}}
\end{definition}
Originally defined by Umegaki \cite{ume:62}. A relevant property of the quantum relative entropy is its monotonicity under CPTP maps, also known as Uhlmann's theorem \cite{uhl:77}.
\begin{theorem}\label{thm:monotonicity-relative-entropy}
		\em Let $\mathcal{H}$ and $\mathcal{K}$ be finite dimensional Hilbert spaces, $\rho_i\in\liouville{}$, i=1,2, density operators with $\mbox{supp}\left(\rho_1\right)\subseteq \mbox{supp}\left(\rho_2\right)$.
		For a CPTP map $\Phi:\mathcal{B}\left(\mathcal{H}\right) \rightarrow\mathcal{B}\left(\mathcal{K}\right) $ the following inequality holds:
		\begin{equation}
		\label{eq:thm:monotonicity-relative-entropy}
		S\left(\rho_1\|\rho_2\right)\geq S\left(\Phi(\rho_1)\|\Phi(\rho_2) \right).
		\end{equation}
\end{theorem}
\begin{corollary}\label{remark:SSAfromUhlmann} The von Neumann entropy is \emph{strong sub-additivite}:
\begin{equation}
S\left( \rho_{ABC}\|\rho_{AB}\otimes\frac{\id_c}{d_C}\right)\geq S\left( \rho_{BC}\|\rho_{B}\otimes\frac{\id_c}{d_C}\right).
\end{equation}
\end{corollary}
\begin{proof}
Observe that setting in \eqref{eq:thm:monotonicity-relative-entropy} $\rho_{1}\rightarrow\rho_{ABC}$, $\rho_{2}\rightarrow\rho_{AB}\otimes\id_C/d_C$ and $\phi(\cdot)\rightarrow\pTr{A}{\cdot}$, we obtain which is equivalent to the non-negativity of the quantum conditional mutual information $I_\rho\left(A:C|B \right)\geq 0 $.
\end{proof}

The following two theorems characterize the case of the equality and they will be the core of the proof of $Lemma~\ref{lemma: QMC - I(A:C|B)}$.
\begin{theorem}\label{thm:ssa=logcondition}\em\textbf{(Theorem 2 at \cite{pet:03})}
		 Let $\Phi:\mathcal{B}\left(\mathcal{H}\right) \rightarrow\mathcal{B}\left(\mathcal{K}\right) $ be a CPTP map and let $\rho_i\in\liouville{}$, i=1,2, and $\cp{i}\in\mathcal{B}\left(\mathcal{K} \right) $ be all invertible density operators. Then, the equality holds in the Uhlmann theorem iff the following equivalent conditions hold:
		\begin{itemize}
			\item[i)] $\phi^\dagger\left(\cp{2}^{it}\cp{1}^{-it} \right)=\rho_{2}^{it}\rho_{1}^{-it} $ $t\in\mathbb{R}$;
			\item[ii)] $\phi^\dagger\left(\log\cp{1}-\log\cp{2}\right)=\log\rho_{1}-\log\rho_{2}$;
		\end{itemize}
		where $(ii)$ is obtained differentiating $(i)$ in t=0.
\end{theorem}
The adjoint map $\phi^\dagger(\cdot)$ is understood with respect to the \emph{Hilbert-Schmidt inner product}.
\begin{theorem}\label{thm:PetzMap}\em\textbf{(Theorem 5.2 at \cite{sut:18})}
A tripartite state $\rho_{ABC}\in\liouville{ABC}$ is a QMC in the order A-B-C iff $\qI{A:C|B}=0$. Furthermore, one can always choose as recovery map the rotated Petz map:
\begin{align}
\label{eq:rotatedPetzmap}
\mathcal{P}_{B\to BC}^{t}(X):=\rho_{BC}^{\frac{1+it}{2}}\left(\rho_B^{-\frac{1+it}{2}}X \rho_B^{-\frac{1-it}{2}}\otimes \id_C\right)\rho_{BC}^{\frac{1-it}{2}}, \text{ for any } X\in \mathcal{B}(\mathcal{H}_{B}),\ t\in\mathbb{R}.
\end{align}
\end{theorem}
%\begin{theorem}\label{thm:equality gammaMap}(Theorem 3 in \cite{pet:03})
%		Assume $\rho_{ABC}$ is invertible. The equality holds in the SSA iff there exists a CPTP unital mapping $\gamma:\mathcal{L}\left(\hilbert{A}\otimes\hilbert{B}\otimes\hilbert{C} \right) \rightarrow\mathcal{L}\left(\hilbert{A}\otimes\hilbert{B} \right)$ s.t. :
%		\begin{itemize}
%			\item[i)]$\tr\left[\rho_{AB}\gamma(x) \right]=\tr\left[ \rho_{ABC}x\right]  $ $\forall x\in\mathcal{L}\left(\hilbert{A}\otimes\hilbert{B}\otimes\hilbert{C}\right)$;
%			\item[ii)]$\gamma(x)=x$ $\forall x\in\liouville{B}$.
%		\end{itemize}
%\end{theorem}
\begin{proof} \emph{of Lemma~\ref{lemma: QMC - I(A:C|B)}}:
	\begin{itemize}
\item[]\textbf{(3 $\Rightarrow$ 1)} This implication comes for free from the definition of QMC. Moreover, the map $\mathcal{P}_{B\rightarrow BC}\left( \cdot\right)$ is clearly CPTP. The complete positivity indeed comes for free from the Hermitianicity of $\rho_B\otimes\id_C/d_C$ and $\rho_{BC}$, then of their square-roots. 

\item[]\textbf{(1 $\Rightarrow$ 3)} Equation~\eqref{eq:rotatedPetzmap} for t=0 gives exactly the Petz map in (3), so the implication comes as corollary of Theorem~\ref{thm:PetzMap}.

\item[]\textbf{(1 $\Leftrightarrow$ 2)} This follows from the statement of Theorem~\ref{thm:PetzMap}.

\item[]\textbf{(2 $\Leftrightarrow$ 4)} It comes as corollary of Theorem~\ref{thm:ssa=logcondition}, using the settings in Corollary~\ref{remark:SSAfromUhlmann}.

	\end{itemize}

\end{proof}
\section{Lemmas for Theorem 2 and 3}\label{ap3}
We need the following Lemmas to derive the proof.
\begin{lemma}
	\label{lemma4}
	\emph{Let $\mathbb{X}=\{A,B,C,D\}$ be the labeling of  parts of a finite dimensional Hilbert space $\hilbert{}$ and $\mathcal{C}=\{\rho_{XY}\in\mathcal{B}\left( \hilbert{XY}\right),\, X,Y\in \mathbb{X}\}$ an admissible set of two-body marginals.
		Assume $\rho_{AB},\rho_{BC}\in\mathcal{C}$ and $\exists\, \tilde{\rho}_{ABC}\in\liouville{ABC}:\, \tilde{\rho}_{ABC}\in\mbox{Comp}(\rho_{AB},\rho_{BC})$ such that $$I_\rho\left( A:C|B\right)=0.$$
		\textbf{a)}$\ $ If the associate graph $\mathcal{G}_\mathcal{C}$ is a chain A-B-C-D (i.e. $\mathcal{C}=\{\rho_{AB},\rho_{BC},\rho_{CD}\}$) then
		\begin{equation}
		\begin{split}
 & \exists\,\tilde{\rho}\in\liouville{}:\, \tilde{\rho}\in\mbox{Comp}\left(\mathcal{C}\right)\ \textnormal{s.t}\ I_\rho\left( A:CD|B\right)=0\quad \textnormal{iff}\\[1.5ex]
 & \exists\, \tilde{\rho}_{BCD}\in\liouville{BCD}:\, \tilde{\rho}_{BCD}\in\mbox{Comp}(\rho_{BC},\rho_{CD})\ \textnormal{s.t}\ I_\rho\left( B:D|C\right)=0.
		\end{split}
		\end{equation}
%		\begin{equation*}
%\mathcal{G}_\mathcal{C}:\quad \textnormal{A\, --\, B\, --\, C --\, D}\quad\quad
%	\exists\ \textnormal{a QMC}\quad\textnormal{B\, --\, C\, --\, D}
%\quad \Leftrightarrow\quad \exists\ \textnormal{a QMC}\quad\textnormal{AB\, -- C\, --\, D}.
%	\end{equation*}
\textbf{b)}$\quad$If the associate graph $\mathcal{G}_\mathcal{C}$ is a star centred in B (i.e. $\mathcal{C}=\{\rho_{AB},\rho_{BC},\rho_{BD}\}$) 
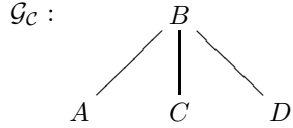
\begin{equation*}
\mathcal{G}_{\mathcal{C}}:\xymatrix{ & B \ar@{-}[dl]\ar@{-}[d]\ar@{-}[dr] & \\
	A &C & D }
	%\quad\quad
%\begin{matrix}
%\quad\\
%\quad\\
%\exists\ \textnormal{a QMC}\quad\textnormal{C\, --\, B\, --\, D}\\
%\quad\\
%\exists\ \textnormal{a QMC}\quad\textnormal{A\, --\, B\, --\, D}
%\end{matrix}\quad \Leftrightarrow\quad \exists\ \textnormal{a QMC}\quad\textnormal{AC\, -- B\, --\, D}.
\end{equation*}
then
\begin{equation}
\exists\,\tilde{\rho}\in\liouville{}:\, \tilde{\rho}\in\mbox{Comp}\left(\mathcal{C}\right)\ \textnormal{s.t}\ I_\rho\left( A:CD|B\right)=0\quad \textnormal{iff}
\end{equation}
		\begin{itemize}
			\item[i)] $\exists\, \tilde{\rho}_{CBD}\in\mbox{Comp}(\rho_{BC},\rho_{BD})$, $\tilde{\rho}_{BCD}\in\liouville{BCD}$ s.t. $I_\rho\left( C:D|B\right)=0$ and
			\item[ii)] $\exists\, \tilde{\rho}_{ABD}\in\mbox{Comp}(\rho_{AB},\rho_{BD})$, $\tilde{\rho}_{ABD}\in\liouville{ABD}$ s.t. $I_\rho\left( A:D|B\right)=0$.
		\end{itemize}
%	\vspace*{-0.5cm}
	In both the cases, $\tilde\rho=\mbox{arg}\underset{{}^{\rho\in\text{Comp}\left(\mathcal{C}\right)} }{\mbox{max}}S(\rho)$ and factorizes over the elements of $\mathcal{C}$ via Petz following a constructive ordering for $\mathcal{C}$.}
	\end{lemma}

\begin{proof} We prove cases a) and b) together, but each direction of the equivalence at a time. We notice than one direction follows easily from the chain rule, we start with that direction
	\textbf{$\mathbf{(\Leftarrow)}$} Recall the chain rule for quantum conditional mutual information:
	\begin{equation}\label{eq:QCMI-chain-rule}
\begin{split}	
\qI{A:X_1,\dots,  X_n | B}=\qI{A:X_1|B}+&\qI{A:X_2|BX_1}+\\&+\dots+\qI{A:X_n|BX_1,\dots, X_{n-1} }
\end{split}
	\end{equation}
and recall that the conditional mutual information is non negative.\\
Case \textbf{a)}
\begin{equation}
\label{3p Chain Rule}
\qI{AB:D|C}=\qI{B:D|C}+\qI{B:D|AC}=0\quad \Rightarrow\quad\qI{B:D|C}=0
\end{equation}
 The case \textbf{b)} is analogous:
 \begin{equation}
 \begin{split}
   \qI{AC:D|B}=\qI{A:D|B}+\qI{A:D|BC}=0\quad \Rightarrow\quad\qI{A:D|B}=0\\
     \qI{AC:D|B}=\qI{C:D|B}+\qI{C:D|AB}=0\quad \Rightarrow\quad\qI{C:D|B}=0
 \end{split}
\end{equation}
\textbf{$\mathbf{(\Rightarrow)}$} \textbf{a)} To prove the other direction of the statement, we show that there exists a $\tilde{\rho}\in\liouville{}$: $\tilde{\rho}\in\mbox{Comp}\left( \tilde{\rho}_{ABC}, \tilde{\rho}_{BCD}\right) $ and QMC on the order AB-C-D.\\
By hypothesis and using Lemma~\ref{lemma: QMC - I(A:C|B)}, the tripartite states can be recovered from two of its two-body marginals using the Petz recovery map:
%\begin{equation}
\begin{align}
\qI{A:C|B}=0\quad \textnormal{iff}\quad\tilde{\rho}_{ABC}=\Petz{BC}{B}{AB}=\Petz{AB}{B}{BC},\label{eq:QMC ABC}\\[1.5ex]
\qI{B:D|C}=0\quad \textnormal{iff}\quad\tilde{\rho}_{BCD}=\Petz{BC}{C}{CD}=\Petz{CD}{C}{BC}.\label{eq:QMC BCD}
\end{align}
%\end{equation}
Using Lemma~\ref{lemma: QMC - comp}, we check the compatibility of the two marginals with the desired QMC showing that the operator $\Theta_{ABCD}:=\tilde{\rho}_{ABC}^\frac{1}{2}\rho_C^{-\frac{1}{2}}\rho_{CD}^\frac{1}{2}$ is normal:
\begin{align}%\label{QMC}
%\begin{split}
\Theta_{ABCD}\Theta_{ABCD}^\dagger&=\tilde{\rho}_{ABC}^\frac{1}{2}\rho_C^{-\frac{1}{2}}\rho_{CD}\rho_C^{-\frac{1}{2}}\tilde{\rho}_{ABC}^\frac{1}{2}\\[1.0ex]
&=\rho_{AB}^\frac{1}{2}\,\rho_{B}^{-\frac{1}{2}}\underset{\tilde{\rho}_{BCD}}{\underbrace{\Petz{BC}{C}{CD}}}\rho_{B}^{-\frac{1}{2}}\rho_{AB}^\frac{1}{2}\label{l9aeq1}\\[1.0ex] 
&=\underbrace{\rho_{AB}^\frac{1}{2}\rho_B^{-\frac{1}{2}}\rho_{CD}^\frac{1}{2}\rho_C^{-\frac{1}{2}}}\rho_{BC}\underbrace{\rho_C^{-\frac{1}{2}}\rho_{CD}^\frac{1}{2}\rho_B^{-\frac{1}{2}}\rho_{AB}^\frac{1}{2}}\label{l9aeq2}\\[1.0ex]
&=\rho_{CD}^\frac{1}{2}\rho_{C}^{-\frac{1}{2}}\underset{\tilde{\rho}_{ABC}}{\underbrace{\Petz{AB}{B}{BC}}}\rho_{C}^{-\frac{1}{2}}\rho_{CD}^\frac{1}{2}\label{l9aeq3}\\
&=\Theta_{ABCD}^\dagger\Theta_{ABCD}.
\end{align}
%\end{equation*}
Equality in Eq.~\eqref{l9aeq1} follows for Eq.~\eqref{eq:QMC ABC} and Lemma~\ref{lemma: QMC - comp}. Equality in Eq.~\eqref{l9aeq2} follows from permuting density operators in different Hilbert spaces.\\[2mm]
\textbf{b)} Similarly to \textbf{a)}, we show that there exists a $\tilde{\rho}\in\liouville{}$: $\tilde{\rho}\in\mbox{Comp}\left( \tilde{\rho}_{ABC}, \tilde{\rho}_{CBD}, \tilde{\rho}_{ABD}\right) $ and QMC on the order AC-B-D. Again, using Lemma~\ref{lemma: QMC - I(A:C|B)}, the tripartite states can be recovered from two of its two-body marginals using the Petz recovery map:
\begin{align}
%\begin{split}
\qI{A:C|B}=0\quad \textnormal{iff}\quad\tilde{\rho}_{ABC}=\Petz{BC}{B}{AB}=\Petz{AB}{B}{BC},\label{eq:starQMCs ABC}\\[1.5ex]
\qI{C:D|B}=0\quad \textnormal{iff}\quad\tilde{\rho}_{CBD}=\Petz{BC}{B}{BD}=\Petz{BC}{B}{BD},\label{eq:starQMCs CBD}\\[1.5ex]
\qI{A:D|B}=0\quad \textnormal{iff}\quad\tilde{\rho}_{ABD}=\Petz{AB}{B}{BD}=\Petz{AB}{B}{BD}\label{eq:starQMCs ABD}.
\end{align}
%\end{equation}
First, using Lemma~\ref{lemma: QMC - comp}, we check the compatibility of the two marginals $\rho_{BD}$ and $\rho_{ABC}$ with the desired QMC showing that the operator $\Theta_{ABCD}:=\tilde{\rho}_{ABC}^\frac{1}{2}\rho_B^{-\frac{1}{2}}\rho_{BD}^\frac{1}{2}$ is normal:
%\begin{equation}
\begin{align}
\Theta_{ABCD}\Theta_{ABCD}^\dagger&=\tilde{\rho}_{ABC}^\frac{1}{2}\rho_B^{-\frac{1}{2}}\rho_{BD}\rho_B^{-\frac{1}{2}}\tilde{\rho}_{ABC}^\frac{1}{2}\\[1.0ex]&
=\rho_{AB}^\frac{1}{2}\,\rho_{B}^{-\frac{1}{2}}\underset{\tilde{\rho}_{BCD}}{\underbrace{\Petz{BC}{B}{BD}}}\rho_{B}^{-\frac{1}{2}}\rho_{AB}^\frac{1}{2}\\[1.0ex]
&=\underset{\tilde{\rho}_{ABD}}{\underbrace{\rho_{AB}^\frac{1}{2}\rho_B^{-\frac{1}{2}}\rho_{BD}^\frac{1}{2}}}\rho_B^{-\frac{1}{2}}\rho_{BC}\rho_B^{-\frac{1}{2}}\underset{\tilde{\rho}_{ABD}}{\underbrace{\rho_{BD}^\frac{1}{2}\rho_B^{-\frac{1}{2}}\rho_{AB}^\frac{1}{2}}}\\[1.0ex]
&=\rho_{BD}^\frac{1}{2}\rho_{B}^{-\frac{1}{2}}\underset{\tilde{\rho}_{ABC}}{\underbrace{\Petz{AB}{B}{BC}}}\rho_{B}^{-\frac{1}{2}}\rho_{BD}^\frac{1}{2}\\
&=\Theta_{ABCD}^\dagger\Theta_{ABCD}.
\end{align}
%\end{equation}
Moreover, the QMC $\tilde{\rho}=\Theta_{ABCD}\Theta_{ABCD}^\dagger$ is in $\mbox{Comp}\left( \tilde{\rho}_{ABC}, \tilde{\rho}_{CBD}, \tilde{\rho}_{ABD}\right)$ by using  Eq.~\eqref{eq:starQMCs ABC}, Eq.~\eqref{eq:starQMCs CBD} and Eq.~\eqref{eq:starQMCs ABD}. 
\end{proof}
The previous lemma can be trivially extended to the n-partite scenario, i.e. to an arbitrary chain and a star:
\begin{corollary}
	\label{corollary: (lemma8):n-parties:chain e star}
	\emph{Let $\mathbb{X}=\{X_1,\dots X_n\}$ be the labeling set of the parts of a finite dimensional Hilbert space $\hilbert{}$ and $\mathcal{C}=\{\rho_{XY}\in\mathcal{B}\left( \hilbert{XY}\right),\, X,Y\in \mathbb{X}\}$ a set of two-body marginals on it classically compatible.
	%	\textbf{a)}$\ $Lets assume $\mathcal{G}_\mathcal{C}$ is a chain $X_1$ - $\dots$ - $X_i$ - $X_{i+1}$ - $\dots$ -$X_n$, i.e. $\mathcal{C}=\{\rho_{X_iX_{i+1}}, i=1,\dots n-1\}$.
	%	\begin{equation}
	%	\exists\,\tilde{\rho}\in\liouville{}:\, \tilde{\rho}\in\mbox{Comp}\left(\mathcal{C}\right)\ \textnormal{quantum Markov network}\ \Leftrightarrow\  I_\rho\left( X_{i-1}:X_{i+1}|X_i\right)=0\ \forall i=2,\dots, n-1.
	%	\end{equation}
	%	\textbf{b)}$\ $Lets 
	Assume $\mathcal{G}_\mathcal{C}$ is a star centred in some $Y\in\mathbb{X}$, i.e. $\mathcal{C}=\{\rho_{X_iY}, i=1,\dots,n-1\}$ then there exists $\tilde{\rho}\in\liouville{}:\, \tilde{\rho}\in\mbox{Comp}\left(\mathcal{C}\right)$ such that $\tilde{\rho}$ is a quantum Markov network
	iff 
	$$I_\rho\left( X_{i}:X_{j}|Y\right)=0\ \forall i\neq j\in{1,\dots,n-1}.$$
	Moreover, $$\tilde\rho=\mbox{arg}\underset{{}^{\rho\in\text{Comp}\left(\mathcal{C}\right)} }{\mbox{max}}S(\rho)$$ and factorizes over the elements of $\mathcal{C}$ via Petz following a constructive ordering for $\mathcal{C}$.}
\end{corollary}
\begin{proof} The proof follows by adding at each step a  node to the setting of Lemma~\ref{lemma4} (case b).  Shortly, consider the constructive ordering for the graph $\mathbb{X}=\{Y, X_1, \dots, X_n\}$. Start from  graph $\mathcal{G}_3$, where $V_3\equiv{Y, X_1, X_2, X_3}$, clearly in this case we are in the situation of Lemma~\ref{lemma4} b), then:
    \begin{equation}
        I_\rho(X_2:X_1|Y)=0\  I_\rho(X_3:X_1|Y)=0\ \Leftrightarrow\ I_\rho(X_3:X_1X_2|Y)=0.
    \end{equation}
    Observe that the two conditions $I_\rho(X_3:X_1X_1|Y)=0$ and $I_\rho(X_2:X_1|Y)=0$ are those required by Theorem~\ref{thm: QMT- Petz fact - log} s.t. there exists a Petz-factorizable d.o. $\rho_3$ over $\mathcal{G}_3$. 
    Next, we add the link $X_4-Y$ to the graph and verify that $I_\rho(X_4:X_1X_2X_3|Y)=0$ also holds. We need to use again Theorem~\ref{thm: QMT- Petz fact - log} to construct a Petz-factorizable $\rho_4$. This condition follows by applying Lemma~\ref{lemma4} b):
    \begin{equation}
    \begin{split}
      I_\rho(X_4:X_1X_2X_3|Y)=0\ \textnormal{iff} \\
      I_\rho(X_3:X_1X_2|Y)=0\ \textnormal{and}\ I_\rho(X_4:X_1X_2|Y)=0.
    \end{split}
    \end{equation}
    Where the first condition is the one we got in the previous step, the second comes from Lemma~\ref{lemma4} b):
    \begin{equation}
     \begin{split}
      I_\rho(X_4:X_1X_2|Y)=0\ \textnormal{iff} \\
      I_\rho(X_4:X_1|Y)=0\ \textnormal{and}\ I_\rho(X_4:X_2|Y)=0.
    \end{split}   
    \end{equation}
    Then, we keep adding nodes and decomposing the next required condition by Theorem~\ref{thm: QMT- Petz fact - log}. We notice that at each step, i.e. every time we add a node, in order to have a Petz decomposable d.o. on the new graph we have to add to the previous set of 3-chains, all the new 3-chains, i.e. the ones that involve the last added node.
\end{proof}
\begin{lemma}\label{lemma:relax} 
	\label{lemma: relaxation Vk->V}	
	\emph{ 
		Let $\rho\in\liouville{}$, where
		$\mathbb{X}=\{X_1,\dots,X_n\}$ and  $\hilbert{}=\bigotimes_{i=1}^n\hilbert{X_i}$,  such that
		$\rho\in\mbox{Comp}\left(\mathcal{C} \right) $ with
%$\mathcal{C}=\left\lbrace \rho_{X_{\ell}Y_{\ell}}\in\mathcal{L}\left(\hilbert{X_\ell Y_\ell} \right):\ X_\ell\in\mathbb{X}, Y_\ell\in\{X_1,\dots, X_{\ell-1}\};\ \ell=3,\dots,n\right\rbrace$ and 
$\mathcal{G}_\mathcal{C}$ a tree (i.e., we work under Assumption~\ref{assump:1}, and take $X_1<\dots<X_n$ the constructive order).
If for some $\ell\leq n$ the following conditions hold
		\begin{equation}
		\label{eq:lemma: relaxation Vk->V.1}
		 I_{\rho_j}\left(X_i:\overline{Y_j}|\,Y_j\right)=0,\ \forall j=3,\dots, \ell,
		\end{equation}
then, by taking any $i$ and $m_i\geq i$ such that
	\begin{equation}
		\label{eq:lemma: relaxation Vk->V.2}
		\mbox{deg}X_i|_{\mathcal{G}_i}=\mbox{deg}X_i|_{\mathcal{G}_{m_i}}\quad \textnormal{and}\quad\mbox{deg}Y_i|_{\mathcal{G}_i}=\mbox{deg}Y_i|_{\mathcal{G}_{m_i}},
		\end{equation}
the following conditions also hold
		\begin{equation}
		 I_{\rho}\left(X_i:V_{r_i}\setminus\{X_i,Y_i\}|\,Y_i\right)=0,\ \forall r_i:\ i\leq r_i\leq m_i\leq \ell,
		\end{equation}
}
\end{lemma}
\begin{proof}
Take that Eq.~\eqref{eq:lemma: relaxation Vk->V.1} with $j=\rho_{r_i}$. 
By Theorem~\ref{thm: QMT- Petz fact - log}, we know that $\rho_{r_i}$ factorizes via Petz over its two-body marginals according to $\mathcal{C}_{r_i}$. 
Then, set $$\Delta_k:=\rho_{X_{k}Y_{k}}^\frac{1}{2}\rho_{Y_k}^{-\frac{1}{2}},$$ it follows that the factorization via Petz can be written as follows:
	\begin{equation}
	\label{eq: rhok}
	\rho_{r_i}=\Delta_{r_i}\Delta_{{r_i}-1}\dots\Delta_{i}\dots\dots\rho_{X_{1}X_{2}}\dots\Delta_i\dots\Delta_{{r_i}-1}\Delta_{r_i},
	\end{equation} 
	where, in general, $[\Delta_i,\Delta_j]\neq 0$. Note that, from the definition of $m_i$ it must be the case that $[\Delta_{r_i},\Delta_{s}]= 0\,\ \forall s: i\leq s\leq {r_i}$. This follows since Eq.~\eqref{eq:lemma: relaxation Vk->V.2} imposes that no more nodes are connected to $X_i$ and $Y_i$ when adding nodes from step $i$ to $m_i$; and therefore, the additional $\Delta_k$'s operate on different Hilbert spaces. %Indeed, the condition~\eqref{eq:lemma: relaxation Vk->V.2} ensures that $\Delta_{j}$ are defined on different Hilbert spaces with respect to $\Delta_{r_i}$, then extended to the whole space tensoring with the identity. This for all $i\leq r_i\leq m_i$, since due to the constructive ordering, $\mbox{deg}X_i|_{\mathcal{G}_i}=\mbox{deg}X_i|_{\mathcal{G}_{m_i}}$ implies also $\mbox{deg}X_i|_{\mathcal{G}_i}=\mbox{deg}X_i|_{\mathcal{G}_{m_i}}=\mbox{deg}X_i|_{\mathcal{G}_{r_i}}$ for all for all $i\leq r_i\leq m_i$.
	Then, Eq.~\eqref{eq: rhok} to be written as:
	\begin{equation}
	\rho_k=\Delta_{i}\dots\Delta_{r_i}\dots\rho_{X_{1}X_{2}}\dots\Delta_{r_i}\dots\Delta_{i}.
	\end{equation}
Now consider a new, but equivalent, constructive ordering $<'$ 
$$X_1<'\dots <'X_{i-1}<'X_{r_i}<'X_{i+1}<'\dots<'X_{r_i-1}<'X_{i}$$, obtained from the order $<$ by exchanging $r_i$ with $i$. By using Theorem~\ref{thm: QMT- Petz fact - log} ii) with the order $<'$, we get in $\mathcal{C}_{r_i}'$ the condition
	\begin{equation}
I_{\rho_{r'_i}}\left(X_{r'_i}:\overline{Y_{r'_i}}|Y_{r'_i}\right)=0.
	\end{equation}
Which for the usual order $<$ can be stated as:
		\begin{equation}
	I_{\rho}\left(X_i:V_{r_i}\backslash\{X_i,Y_i\}|Y_i\right)=0.
	\end{equation}
The latter equality is valid for all $r_i:\ i\leq r_i\leq m_i\leq \ell$, since the only property used was the fact that $\mbox{deg}X_i|_{\mathcal{G}_i}=\mbox{deg}X_i|_{\mathcal{G}_{r_i}}$.
\end{proof}

\end{document}